\documentclass[a4paper, 11pt,  fleqn]{article}
\usepackage{hyperref}
\usepackage{fullpage}
\usepackage{graphics}

\usepackage{algorithm, algorithmic}
\usepackage{amsmath}
\usepackage{amssymb}
\usepackage{amsfonts}
\usepackage{braket}
\usepackage{comment}
\usepackage{mathtools}
\usepackage{amsthm}
\usepackage{color}
\usepackage{comment}
\usepackage{cite}

\newtheorem{theorem}{Theorem}
\newtheorem{lemma}[theorem]{Lemma}

\newtheorem{proposition}[theorem]{Proposition}

\theoremstyle{definition}

\newtheorem{remark}[theorem]{Remark}

\usepackage[appendix=append, bibliography=common]{apxproof}

\newcommand{\switch}[2]{#2} 

\newcommand{\I}{{\cal I}}
\newcommand{\opt}{{\rm OPT}}
\newcommand{\optvc}{\tau}
\newcommand{\wst}{{\rm WST}}
\newcommand{\alg}{{\rm ALG}}
\newcommand{\app}{{\rm APPROX}}

\newcommand{\state}{{\sf state}}

\title{
Incomplete List Setting of the Hospitals/Residents Problem with Maximally Satisfying Lower Quotas
\thanks{This work was partially supported by the joint project of Kyoto University and Toyota Motor Corporation, titled ``Advanced Mathematical Science for Mobility Society''.}} 

\author{
Kazuhisa Makino\thanks{Research Institute for Mathematical Sciences, Kyoto University, Kyoto
606-8502, Japan., E-mail: {\tt makino@kurims.kyoto-u.ac.jp}, Supported by JSPS KAKENHI Grant Numbers JP20H05967, JP19K22841, and JP20H00609.}
\and
Shuichi Miyazaki\thanks{Graduate School of Information Science, University of Hyogo
8-2-1 Gakuennishi-machi Nishi-ku Kobe, Hyogo, 651-2197, Japan, E-mail: {\tt shuichi@sis.u-hyogo.ac.jp}, Supported by JSPS KAKENHI Grant Number JP20K11677.}
\and
Yu Yokoi\thanks{Principles of Informatics Research Division, National Institute of Informatics, Tokyo 101-8430, Japan, E-mail: {\tt yokoi@nii.ac.jp}, Supported by JSPS KAKENHI Grant Number JP18K18004 and JST, PRESTO Grant Number JPMJPR212B.}}

\begin{document}
\maketitle  
\begin{abstract}
To mitigate the imbalance in the number of assignees in the Hospitals/Residents problem, Goko et al.~[Goko et al., Maximally Satisfying Lower Quotas in the Hospitals/Residents Problem with Ties, Proc.~STACS 2022, pp.~31:1--31:20] studied the Hospitals/Residents problem with lower quotas whose goal is to find a stable matching that satisfies lower quotas as much as possible.
In their paper, preference lists are assumed to be complete, that is, the preference list of each resident (resp., hospital) is assumed to contain all the hospitals (resp., residents).

In this paper, we study a more general model where preference lists may be incomplete.
For four natural scenarios, we obtain maximum gaps of the best and worst solutions, approximability results, and inapproximability results.
\end{abstract}
%
%
%

\section{Introduction}\label{sec:intro}
The Hospitals/Residents problem is a many-to-one generalization of the stable marriage model of Gale and Shapley \cite{10.2307/2312726}, which is widely used as a base for many real-world assignment systems, such as assigning residents to hospitals \cite{RePEc:ucp:jpolec:v:92:y:1984:i:6:p:991-1016} and students to schools \cite{10.1257/000282805774670167, 10.1257/000282805774669637}.
In applications, one of the major drawbacks is the imbalance in the number of assignees.
For example, in the hospitals-residents assignment, hospitals in urban areas are popular and are likely to receive many residents, while those in rural areas suffer from a shortage of doctors.
To cope with this problem, several approaches 
have been proposed.
One of them is to introduce regional caps, where hospitals in the same district are put into the same set, called a region, and each region as well as each hospital has an upper quota \cite{RePEc:sip:dpaper:10-011}.
Another example is to let each hospital $h$ declare not only an upper quota $u(h)$ but also a lower quota $\ell(h)$, and require the number of assignees for each hospital to be between its upper and lower quotas \cite{DBLP:journals/tcs/BiroFIM10, DBLP:journals/algorithmica/HamadaIM16}.
In general, there does not exist a stable matching satisfying all the upper and lower quotas, so Bir{\'{o}} et al.~\cite{DBLP:journals/tcs/BiroFIM10} allows some hospitals to be closed (i.e., to accept no resident) while Hamada et al.~\cite{DBLP:journals/algorithmica/HamadaIM16} allows a matching to be unstable and tries to minimize the number of blocking pairs.
However, in reality, it is too pessimistic to regard lower quotas as hard constraints.
In some cases, fewer residents than a lower quota may still be able to keep a hospital in operation, by sacrificing service level to some degree.

Taking this observation into account, recently Goko et al.~\cite{goko_et_al:LIPIcs.STACS.2022.31} (including the authors of the present paper) investigated an optimization problem that considers lower quotas as soft constraints.
The problem is called {\em the Hospitals/Residents problem with Ties to Maximally Satisfy Lower Quotas} (HRT-MSLQ for short).
In this problem, preference lists of agents may contain ties and the {\em satisfaction ratio} (or the {\em score}) of a hospital $h$ in a stable matching $M$ is defined as $s_{M}(h) = \min\{ 1, \frac{|M(h)|}{\ell(h)} \}$, where $|M(h)|$ is the number of residents assigned to $h$ in $M$ and $s_M(h)=1$ if $\ell(h)=0$.
That is, each hospital $h$ is assumed to have $\ell(h)$ seats corresponding to the lower quota, and the satisfaction ratio reflects the fraction of occupied seats among them.
The {\em total satisfaction ratio} (also called {\em the score}) of a stable matching $M$ is the sum of the scores of all hospitals, and the goal of HRT-MSLQ is to maximize it among all the stable matchings.
(See  Section~\ref{sec:definition} for the formal definition.) 
Note that if the preference lists of agents are strict, the problem is trivial because each hospital receives the same number of residents in any stable matching due to the {\em rural hospitals theorem} \cite{DBLP:journals/dam/GaleS85,RePEc:ucp:jpolec:v:92:y:1984:i:6:p:991-1016,RePEc:ecm:emetrp:v:54:y:1986:i:2:p:425-27}, and hence all the stable matchings have the same score.
As the problem HRT-MSLQ turns out to be NP-hard in general, Goko et al. proposed an approximation algorithm called {\sc Double Proposal} and compared the worst approximation factors of this algorithm and a baseline algorithm that first breaks ties arbitrarily and then applies the ordinary Gale--Shapley algorithm.
They showed tight approximation factors of these two algorithms for the following four scenarios: (i) \emph{general model}, which consists of all problem instances, (ii) \emph{uniform model}, in which all hospitals have the same upper and lower quotas, (iii) \emph{marriage model}, in which each hospital has an upper quota of $1$ and a lower quota of either $0$ or $1$, and (iv) resident-side master list model (\emph{$R$-side ML model}), in which there is a {\em master preference list} over hospitals and a preference list of every resident coincides with it.
Their algorithm {\sc Double Proposal} solves the problem exactly for instances of \emph{$R$-side ML model} and, for each of the other three models, it attains a better approximation factor than that of the baseline algorithm. 

Throughout Goko et al.'s paper \cite{goko_et_al:LIPIcs.STACS.2022.31}, it is assumed that all preference lists are complete. That is, the preference list of each resident (resp., hospital) is assumed to contain all the hospitals (resp., residents).
Under this assumption, all stable matchings have the same cardinality, and hence the objective value is affected only by the {\em balance} of the numbers of residents assigned to hospitals.
Therefore, to obtain a better objective value, we can focus on rectifying the imbalance. 

In reality, however, preference lists of agents can be incomplete. 
In real medical residency matching schemes, there are huge numbers of residents and hospitals participating, and the length of preference lists of residents tend to be much smaller than the number of hospitals. 
This motivates us to study a more natural but challenging incomplete list setting.
When preference lists are incomplete and contain ties, the problem becomes harder since the stable matchings vary on their cardinalities and hence the objective value depends not only on balance but also on size.
In fact, a special case of our problem where $\ell(h)=u(h)=1$ for every hospital $h$ reduces to MAX-SMTI \cite{DBLP:journals/tcs/ManloveIIMM02}, a well-studied NP-hard problem that asks to find a maximum size stable matching in the stable marriage model with ties and incomplete lists.

For instances with incomplete lists, the algorithm {\sc Double Proposal} in \cite{goko_et_al:LIPIcs.STACS.2022.31} is still well-defined and finds a stable matching in linear time.
However, the analyses of the approximation factors do not work anymore.
Indeed, the algorithm fails to attain the approximation factors shown there for this generalized setting.
In this paper, we propose a new algorithm {\sc Triple Proposal}, which is based on {\sc Double Proposal} and equipped with additional operations to increase the number of matched residents, which leads to a larger objective value.
As mentioned above, our problem is a common generalization of the problem of Goko et al.~\cite{goko_et_al:LIPIcs.STACS.2022.31} and MAX-SMTI.
Our algorithm {\sc Triple Proposal} can be regarded as a natural integration of {\sc Double Proposal} and Kir\'aly's algorithm \cite{DBLP:journals/algorithms/Kiraly13} for MAX-SMTI.

Our results are summarized in Table \ref{table1}, where $n$ denotes the number of residents in an input instance.
For each of the four models mentioned above, the first row shows the maximum gap, i.e., the ratio of objective values of the best and worst stable matchings. The maximum gap can be regarded as the worst-case approximation factor of the baseline algorithm. (This fact follows from the rural hospitals theorem and is proved formally in \cite{Goko}, which is a full version of \cite{goko_et_al:LIPIcs.STACS.2022.31}.)
The second row shows lower bounds of the approximation factors of {\sc Double Proposal} in our incomplete list setting.
It can be observed that in the uniform and marriage models, {\sc Double Proposal}  performs as badly as the baseline algorithm, in contrast to the fact that it performs better than the baseline algorithm in all of the four models in the complete list setting \cite{goko_et_al:LIPIcs.STACS.2022.31}.
The third row gives the approximation factors of {\sc Triple Proposal}, which are the main results of this paper.
Here, $\phi(n)$ is a function on $n$ that is approximately $\frac{n+2}{3}$; its formal definition is given just before Theorem \ref{thm:general-approximable} in Section~\ref{subsec:general}.
Note that, in all four models, the approximation factors are improved from {\sc Double Proposal}.
Finally, the last row shows inapproximability results under the assumption that P$\neq$NP.
It is worth emphasizing that, as a byproduct of our inapproximability result, lower bounds on the approximation factor of MAX-SMTI (under P$\neq$NP) can be improved (see Remark~\ref{rem:SMTI}).

Note that, in the incomplete list setting, the $R$-side ML model means that each resident's preference list is obtained by deleting unacceptable hospitals from the master preference list.
This model shows a remarkable difference between complete list setting and incomplete list setting: the model is solved exactly by {\sc Double Proposal} in the former case \cite{goko_et_al:LIPIcs.STACS.2022.31}, while it has the same inapproximability as the general model in the latter case.

\renewcommand\arraystretch{0.9}
\begin{center}
\begin{table}[htbp]
  \centering
    \begin{tabular}{|c| c| c| c| c|} \hline
   & General  &\ Uniform \ &\  Marriage\ & \ $R$-side ML\  \\\hline\hline
   \begin{tabular}{c}~\vspace{-2mm}\\
   Maximum gap
   \vspace{-0.5mm}\\
   {\scriptsize (i.e.,  approx. factor of} \\[-.1cm]
   {\scriptsize arbitrary tie-break+GS)}
   \end{tabular}

   &\begin{tabular}{c}~\vspace{-1mm}\\
   $n+1$ \\
   \tiny(Thm.\ref{thm:maxgap_gen})
   \end{tabular}
   &\begin{tabular}{c}~\vspace{-1mm}\\
   $\theta+1$ \\
   \tiny(Thm.\ref{thm:maxgap_uni})
   \end{tabular}
   &\begin{tabular}{c}~\vspace{-1mm}\\
   $2$ \\
   \tiny(Thm.\ref{thm:maxgap_mar})
   \end{tabular}
   &\begin{tabular}{c}~\vspace{-1mm}\\
   $n+1$ \\
   \tiny(Thm.\ref{thm:maxgap_gen})
   \end{tabular}
\\ \hline 
   \begin{tabular}{c}~\vspace{-3mm}\\
   \hspace{-2mm}LB of approx. factor\hspace{-3mm}\\
   \hspace{-2mm}of {\sc Double Proposal}\hspace{-3mm}
   \end{tabular}
   &\begin{tabular}{c}~\vspace{-1mm}\\
   $\lfloor \frac{n}{2} \rfloor+1$ \\
   \tiny(Rem.\ref{rem:double_proposal}+Thm.\ref{thm:maxgap_gen_H})
   \end{tabular}
   &\begin{tabular}{c}~\vspace{-1mm}\\
   $\theta+1$ \\
   \tiny(Rem.\ref{rem:double_proposal}+Thm.\ref{thm:maxgap_uni})
   \end{tabular}
   &\begin{tabular}{c}~\vspace{-1mm}\\
   $2$ \\
   \tiny(Rem.\ref{rem:double_proposal}+Thm.\ref{thm:maxgap_mar})
   \end{tabular}
   &\begin{tabular}{c}~\vspace{-1mm}\\
   $\lfloor \frac{n}{2} \rfloor+1$ \\
   \tiny(Rem.\ref{rem:double_proposal}+Thm.\ref{thm:maxgap_gen_H})
   \end{tabular}
\\ \hline 
   \begin{tabular}{c}~\vspace{-3mm}\\
   Approx. factor of \\
   {\sc Triple Proposal}
   \end{tabular}
   &\begin{tabular}{c}~\vspace{-1mm}\\
   $\phi(n)~(\sim \frac{n+2}{3})$ \\
   \tiny(Thm.\ref{thm:general-approximable})
   \end{tabular}
   &\begin{tabular}{c}~\vspace{-1mm}\\
   $\frac{\theta}{2}+1$ \\
   \tiny(Thm.\ref{thm:IL-uniform-alg})
   \end{tabular}
   &\begin{tabular}{c}~\vspace{-1mm}\\
   $1.5$ \\
   \tiny(Thm.\ref{thm:approx_mar})
   \end{tabular}
   &\begin{tabular}{c}~\vspace{-1mm}\\
   $\phi(n)$ \\
   \tiny(Thm.\ref{thm:general-approximable})
   \end{tabular}
\\ \hline 
   \ Inapproximability
   &\begin{tabular}{c}~\vspace{-1mm}\\
   $n^{\frac{1}{4}-\epsilon}$\, \\
   \tiny(\cite{goko_et_al:LIPIcs.STACS.2022.31}~)
   \end{tabular}
   &\begin{tabular}{c}~\vspace{-1mm}\\
   \hspace{-2mm}$\frac{2-\sqrt{2}}{3}\theta+1-\epsilon \theta^2\hspace{-2mm}$\\
   \tiny(Thm.\ref{thm:inapprox-uniform})
   \end{tabular}
   &\begin{tabular}{c}~\vspace{-1mm}\\
   $\frac{5-\sqrt{2}}{3}-\epsilon$\\
   \tiny(Prop.\ref{prop:marriage-inapprox})
   \end{tabular}
   &\begin{tabular}{c}~\vspace{-1mm}\\
   $n^{\frac{1}{4}-\epsilon}$ \\
   \tiny(\cite{goko_et_al:LIPIcs.STACS.2022.31}+ Rem.\ref{rem:ML-inapprox})
   \end{tabular}
  \\ \hline
  \end{tabular}
\smallskip 
  \caption{\small The maximum gap, a lower bound of the approximation factor of {\sc Double Proposal}, the approximation factor of {\sc Triple Proposal}, and inapproximability (under P$\neq$NP) for four models of HRT-MSLQ.
  In the uniform model, $\theta = u/\ell$, where $u$ and $\ell$ are respectively an upper quota and a lower quota common to all the hospitals.}
  \label{table1}
\end{table}
\end{center}
\medskip

\section{Problem Definition}\label{sec:definition}
Let $R = \{ r_{1}, r_{2}, \ldots, r_{n} \}$ be a set of residents and $H = \{ h_{1}, h_{2}, \ldots, h_{m} \}$ be a set of hospitals.
Each hospital $h$ has a lower quota $\ell(h)$ and an upper quota $u(h)$ such that $\ell(h) \leq u(h)\leq n$.
We sometimes denote $h$'s quota pair as $[\ell(h), u(h)]$ for simplicity.
Each resident (resp., hospital) has a preference list over hospitals (resp., residents), which may be incomplete and may contain ties.
A hospital $h$ is {\em acceptable} to a resident $r$ if $h$ appears in $r$'s list.
Similarly,  $r$ is {\em acceptable} to $h$ if $r$ appears in $h$'s list.
A resident-hospital pair $(r,h)$ is called {\em acceptable} if $r$ and $h$ are acceptable to each other.
We denote by $E$ the set of all acceptable pairs $(r,h)$.
In this paper, we assume without loss of generality that acceptability is mutual, i.e., $r$ is acceptable to $h$ if and only if $h$ is acceptable to $r$.

In this paper, a preference list is denoted by one row, from left to right according to the preference order.
Two or more agents with equal preference is included in parentheses.
For example, ``$r_{1}$: \ $h_{3}$ \  ( \ $h_{1}$ \ $h_{6}$ \ ) \ $h_{4}$''
is a preference list of resident $r_1$ such that $h_{3}$ is the top choice, $h_{1}$ and $h_{6}$ are the  second choice with equal preference, and $h_{4}$ is the last choice. Hospitals not appearing in the list are unacceptable to $r_1$.

If hospitals $h_{i}$ and $h_{j}$ are both acceptable to a resident $r$ and $r$ prefers $h_{i}$ to $h_{j}$, we write $h_{i} \succ_{r} h_{j}$.
We also write $h\succ_{r} \varnothing$ for any hospital $h$ acceptable to $r$, where $\varnothing$ means $r$ being unmatched.
If $r$ is indifferent between acceptable hospitals $h_{i}$ and $h_{j}$ (including the case that $h_{i} = h_{j}$), we write $h_{i} =_{r} h_{j}$.
We write $h_{i} \succeq_{r} h_{j}$ to mean that $h_{i} \succ_{r} h_{j}$ or $h_{i} =_{r} h_{j}$ holds.
We use the same notations for the preference of each hospital.

An {\em assignment} is a subset of the set $E$ of acceptable pairs.
For an assignment $M$ and a resident $r$, let $M(r)$ be the set of hospitals $h$ such that $(r, h) \in M$.
Similarly, for a hospital $h$, let $M(h)$ be the set of residents $r$ such that $(r, h) \in M$.
An assignment $M$ is called a {\em matching} if $|M(r)| \leq 1$ for each resident $r$ and $|M(h)| \leq u(h)$ for each hospital $h$.
For a matching $M$, a resident $r$ is called {\em matched} if $|M(r)|=1$ and {\em unmatched} otherwise.
If $(r,h) \in M$, we say that $r$ is {\em assigned to} $h$ and $h$ is {\em assigned} $r$.  
For a matching $M$, we abuse the notation $M(r)$ to denote a unique hospital to which $r$ is assigned, if any.  
In addition, we denote  $M(r)=\varnothing$ to mean that $r$ is unmatched in $M$.
A hospital $h$ is called {\em deficient} if  $|M(h)| < \ell(h)$ and {\em sufficient} if $\ell(h) \leq |M(h)| \leq u(h)$.
Additionally, a hospital $h$ is called {\em full} if $|M(h)| = u(h)$ and {\em undersubscribed} otherwise.

An acceptable pair $(r,h)\in E$ is called a {\em blocking pair} for a matching $M$ (or we say that $(r, h)$ {\em blocks} $M$) if (i) $r$ is either unmatched in $M$ or prefers $h$ to $M(r)$ and (ii) $h$ is either undersubscribed in $M$ or prefers $r$ to at least one resident in $M(h)$.
A matching is called {\em stable} if it admits no blocking pair. 

Recall from Section \ref{sec:intro} that the {\em satisfaction ratio} (also called {\em the score}) of a hospital $h$ in a matching $M$ is defined by $s_{M}(h) = \min\{ 1, \frac{|M(h)|}{\ell(h)} \}$, where we define $s_M(h)=1$ if $\ell(h)=0$. 
The {\em total satisfaction ratio} (also called {\em the score}) of a matching $M$ is the sum of the scores of all hospitals, that is, 
$s(M) = \sum_{h \in H} s_{M}(h)$.
The Hospitals/Residents problem with Ties to Maximally Satisfy Lower Quotas, denoted by {\em HRT-MSLQ}, is to find a stable matching $M$ that maximizes the score $s(M)$.
(In Goko et al.\cite{goko_et_al:LIPIcs.STACS.2022.31}, HRT-MSLQ is formulated for the complete list setting. We use the same name HRT-MSLQ to mean the generalized setting with incomplete lists.)


\section{Algorithm}\label{sec:algorithm}
In this section, we present our algorithm {\sc Triple Proposal} 
for HRT-MSLQ along with a few of its basic properties.
{\sc Triple Proposal} is regarded as a generalization of the algorithm {\sc Double Proposal},
which was designed for a special case of HRT-MSLQ with complete preference lists \cite{goko_et_al:LIPIcs.STACS.2022.31}.

We first briefly explain {\sc Double Proposal}.
It is based on the resident-oriented Gale--Shapley algorithm
but allows each resident $r$ to make proposals twice to each hospital. 
{\sc Double Proposal} starts with an empty matching $M\coloneqq \emptyset$ and repeatedly updates $M$ by a proposal-acceptance/rejection process.
In each iteration, the algorithm takes a currently unassigned resident $r$ and lets her propose to the hospital at the top of her current list.
If the head of $r$'s preference list is a tie when $r$ makes a proposal, 
then the hospitals to which $r$ has not proposed yet are prioritized.
If there are more than one hospital with the same priority, one with the smallest $\ell(h)$ is prioritized.
For the acceptance/rejection decision of $h$, we use $\ell(h)$ as a dummy upper quota. 
Whenever $|M(h)|<\ell(h)$, a hospital $h$ accepts any proposal. 
If $h$ receives a new proposal from $r$ when $|M(h)|\geq \ell(h)$, then $h$ checks whether there is a resident in $M(h)\cup\{r\}$ who has not been rejected by $h$ so far. 
If such a resident exists, $h$ rejects that resident regardless of the preference of $h$ (at this point, the rejected resident does not delete $h$ from her list).
If there is no such resident, we apply the usual acceptance/rejection operation, i.e., $h$ accepts $r$ if $|M(h)|<u(h)$ and otherwise $h$ replaces the worst resident $r'$ in $M(h)$ with $r$ (at this point, $r'$ deletes $h$ from her list). 
Roughly speaking, the first proposals are used to prioritize deficient hospitals and
the second proposals are used to guarantee the stability.

In the setting in Goko et al.~\cite{goko_et_al:LIPIcs.STACS.2022.31}, it is assumed that the preference lists of all agents are complete and $|R|\leq \sum_{h\in H}u(h)$ holds.
With these assumptions, all residents are matched in $M$ at the end of {\sc Double Proposal}.
In out setting, however, some residents are left unmatched at the end, after proposing twice to each hospital.
{\sc Triple Proposal} gives the second round of proposals for those unmatched residents, in which a resident has the {\em third} chance for the proposal to the same hospital.
This idea originates from the algorithm of Kir\'aly \cite{DBLP:journals/algorithms/Kiraly13} for another NP-hard problem MAX-SMTI.
In {\sc Triple Proposal}, each resident $r$ is associated with a value $\state(r)\in \{0,1,2\}$, which is initialized as $\state(r)=0$ at the beginning.
Each $r$ behaves as {\sc Double Proposal} until her list becomes empty for the first time (i.e., until $r$ is rejected twice by all hospitals on her list).
When her list becomes empty, $\state(r)$ turns from $0$ to $1$ and her list is recovered.
Note that at this point, every hospital $h$ in $r$'s list is full with residents that are no worse than $r$ and have been rejected by $h$ at least once, 
because $r$ was rejected by $h$ at the second proposal.
When $r$ with $\state(r)=1$ proposes to a hospital $h$ and $M(h)$ contains a resident $r'$ such that $r=_{h} r'$ and $\state(r')=0$, 
then $h$ replaces $r'$ with $r$. 
If $r$ exhausts her list for the second time, then $\state(r)$ turns from $1$ to $2$, and $r$ becomes inactive.

Formally, our algorithm {\sc Triple Proposal} is described in Algorithm~\ref{alg2}. 
For convenience, in the preference list, a hospital $h$ not included in a tie is regarded as a tie consisting of $h$ only.

\begin{algorithm}[htb]
	\caption{\sc ~Triple Proposal}
\label{alg2}
\begin{algorithmic}[1]
	\REQUIRE An instance $I$ of HRT-MSLQ where each $h\in H$ has quotas $[\ell(h),u(h)]$.
	\ENSURE A stable matching $M$.
	\vspace{1mm}
	\STATE $M\coloneqq \emptyset$. 
	\WHILE{there is an unmatched resident $r$ with nonempty list and  $\state(r)\leq 1$}\label{main_while}
	\STATE Let $r$ be such a resident and $T$ be the top tie of $r$'s list.\label{unmatched}
	\IF{$T$ contains a hospital to which $r$ has not proposed yet}\label{propose}
	\STATE Let $h$ be such a hospital with minimum $\ell(h)$. \label{proposal1}
	\ELSE
	\STATE Let $h$ be a hospital with minimum $\ell(h)$ in $T$. \label{proposal2}
	\ENDIF
	\IF{$|M(h)|<\ell(h)$}\label{XXX1}
	\STATE Let $M\coloneqq M\cup\{(r,h)\}$.\label{update1}
	\ELSIF{there is a resident in $M(h)\cup\{r\}$ who has not been rejected by $h$}\label{reject}
	\STATE Let $r'$ be such a resident (possibly $r'=r$).\label{chosen1}
	\STATE Let $M\coloneqq (M\cup \{(r,h)\})\setminus\{(r',h)\}$.\label{reject-end}
	\ELSIF{$|M(h)|<u(h)$}\label{XXX2}
	\STATE $M\coloneqq M\cup\{(r,h)\}$.\label{update3}
	\ELSE[i.e., when $|M(h)|=u(h)$ and all residents in $M(h)\cup\{r\}$ have been rejected by $h$ at least once]\label{u-full}
	\STATE Let $r'$ be any resident worst in $M(h)\cup\{r\}$ for $h$ (possibly $r'=r$).
	If there are multiple worst residents, take anyone with the minimum $\state$ value.\label{chosen2}
	\STATE Let $M\coloneqq (M\cup \{(r,h)\})\setminus\{(r',h)\}$.\label{update4}
	\STATE Delete $h$ from $r'$'s list.
	\ENDIF
	\IF{$r'$'s list becomes empty}
	\STATE Increment $\state(r)$ by $1$ and recover the original $r'$'s list.
	\ENDIF
	\ENDWHILE
	\STATE Output $M$ and halt.
\end{algorithmic}
\end{algorithm}

We say that a resident is {\em rejected} by a hospital $h$ if she is chosen as $r'$ 
in Lines~\ref{chosen1} or \ref{chosen2}.
Note that the conditions in Lines~\ref{XXX1}, \ref{reject}, and \ref{XXX2} never hold for $r$ with $\state(r)=1$ because such $r$ has already been rejected by all hospitals in her list twice. 
Therefore, whenever $r$ with $\state(r)=1$ is rejected by $h$, it must be at line  \ref{chosen2}, so $r$ deletes $h$.
Hence, each resident proposes to each hospital at most three times (at most twice with $\state(r)=0$ and at most once with $\state(r)=1$).

We first provide two basic properties of {\sc Triple Proposal}, which are shown for {\sc Double Proposal} in \cite[Lemmas 1, 2]{goko_et_al:LIPIcs.STACS.2022.31}.
The first one states that the algorithm finds a stable matching in linear time.
\begin{lemma}\label{lem:stability}
Algorithm {\sc Triple Proposal} runs in linear time and outputs a stable matching.
\end{lemma}
\begin{proof}
Clearly, the size of the input is $O(|R||H|)$.
As each resident proposes to each hospital at most three times, 
the while loop is iterated at most $3|R||H|$ times.
As residents prioritize hospitals with smaller $\ell(h)$ at Lines~\ref{proposal1} and \ref{proposal2}, we need to sort hospitals in each tie in an increasing order of the values of $\ell$. Since $\ell$ has only $|R|+1$ possible values ($0,1,2,\dots,n$),
the required sorting can be done in $O(|R||H|)$ time as a preprocessing step using a method like bucket sort. Thus, the running time of {\sc Triple Proposal} is $O(|R||H|)$.

Suppose, to the contrary, that $M$ is not stable, i.e., there is a pair $(r,h)\in E$ such that (i) $r$ prefers $h$ to $M(r)$ (which may be $\varnothing$) and (ii) $h$ is either undersubscribed or prefers $r$ to at least one resident in $M(h)$.
By the algorithm, (i) implies that $r$ is rejected by $h$ at least twice.
Just after the second rejection, $h$ is full, and all residents in $M(h)$ have once been rejected by $h$ and are no worse than $r$ for $h$. 
Since $M(h)$ is monotonically improving for $h$, at the end of the algorithm $h$ is still full and no resident in $M(h)$ is worse than $r$, which contradicts (ii).
\switch{\qed}{}
\end{proof}

In addition to the stability, the output of {\sc Triple Proposal} satisfies the following property, 
which is used in the approximation factor analysis in Section~\ref{sec:approximation}.
\begin{lemma}\label{lem:property}
Let $M$ be the output of {\sc Triple Proposal}, $r$ be a resident, and $h$ and $h'$ be hospitals such that $h=_{r} h'$ and $M(r)=h$.
Then, we have the following:
\begin{enumerate}
\item[\rm (i)] If $\ell(h)> \ell(h')$, then $|M(h')|\geq \ell(h')$.
\item[\rm (ii)] If $|M(h)|> \ell(h)$, then $|M(h')|\geq \ell(h')$.
\end{enumerate}
\end{lemma}
\begin{proof}
(i) Since $h=_{r} h'$, $\ell(h)>\ell(h')$, and $r$ is assigned to $h$ in $M$, 
the definition of the algorithm (Lines~\ref{propose}, \ref{proposal1}, and \ref{proposal2})
implies that $r$ proposed to $h'$ and was rejected by $h'$ before she proposes to $h$.
Just after this rejection occurred, $|M(h')|\geq \ell(h')$ holds. 
Since $|M(h')|$ is monotonically increasing, we also have $|M(h')|\geq \ell(h')$ at the end.

(ii) Since $|M(h)|>\ell(h)$, the value of $|M(h)|$ changes from $\ell(h)$ to $\ell(h)+1$ at some moment of the algorithm.
By Line~\ref{reject},  
at any point after this, $M(h)$ consists only of residents who have once been rejected by $h$.
Since $M(r)=h$ for the output $M$, at some moment $r$ must have made the second proposal to $h$. 
By Line~\ref{propose} of the algorithm, $h=_r h'$ implies that $r$ has been rejected by $h'$ at least once, which implies that $|M(h')|\geq \ell(h')$ at this moment and also at the end.
\switch{\qed}{}
\end{proof}

Since the above two properties hold also for {\sc Double Proposal} \cite{goko_et_al:LIPIcs.STACS.2022.31}, they are insufficient to obtain the approximation factors shown in the third row of Table~\ref{table1}.
The main advantage of {\sc Triple Proposal} is that we can prohibit length-3 augmenting paths, which are defined as follows.

For a stable matching $M$ in an HRT-MSLQ instance, a {\em length-3 $M$-augmenting path} is a sequence $(r_1, h_1, r_2, h_2)$ 
of residents $r_1$, $r_2$ and hospitals $h_1$, $h_2$ satisfying the following conditions:
\begin{itemize}
\item  $M(r_1)=\varnothing$,
\item $(r_1,h_1), (r_2, h_2)\in E\setminus M$,~ $(r_2,h_1) \in M$, and
\item $|M(h_2)|< \ell(h_2)$.
\end{itemize}
The output of {\sc Triple Proposal} satisfies the following property,
which can be regarded as a generalization of the property of the output of Kir\'aly's algorithm~\cite{DBLP:journals/algorithms/Kiraly13} for MAX-SMTI.
\begin{lemma}\label{lem:property2}
Let $M$ be an output of {\sc Triple Proposal} and $N$ be an arbitrary stable matching.
Then, there is no length-3 $M$-augmenting path $(r_1, h_1, r_2, h_2)$ such that $(r_1,h_1), (r_2, h_2)\in N$.
\end{lemma}
\begin{proof}
Suppose, to the contrary, that there is such an augmenting path.
Then, $(r_1,h_1), (r_2, h_2)\in N\setminus M$ and $(r_2, h_1)\in M\setminus N$.
Because $M(r_1)=\varnothing$ while $(r_1, h_1)\in E$, resident $r_1$ was rejected by $h_1$ three times in the algorithm.
Just after the third rejection, $h_1$ is full with residents each of whom is either 
(i) better than $r_1$ and rejected by $h_1$ at least once 
or (ii) tied  with $r_1$ in $h_1$'s list and rejected by $h_1$ twice.
Since the assignment of $h_1$ is monotonically improving, this also holds for the output $M$.

In case $r_2\in M(h_1)$ satisfies (i), the stability of $N$ implies $h_2\succeq_{r_2} h_1$.
Because the stability of $M$ with $|M(h_2)|<\ell(h_2)$ implies $h_1\succeq_{r_2} h_2$, we have $h_2=_{r_2} h_1$.
Since $r_2$ is assigned to $h_1$ at her second or third proposal, $h_2=_{r_2} h_1$ implies 
that $r_2$ proposed to and was rejected by $h_2$ at least once. This implies $|M(h_2)|\geq \ell(h_2)$, 
which contradicts our assumption $|M(h_2)|<\ell(h_2)$.

In case $r_2$ satisfies (ii), $r_2$ was rejected at least twice by all hospitals in her list.
This implies that $h_2$ is full in $M$, which contradicts $|M(h_2)|<\ell(h_2)$.
\switch{\qed}{}\end{proof}

\begin{remark}\label{rem:3.1}
It is shown in Goko et al.~\cite{goko_et_al:LIPIcs.STACS.2022.31} that {\sc Double Proposal} is strategy-proof for residents.
In case where (i) the preference lists are all complete or (ii) the preference lists of hospitals are strict,
our algorithm {\sc Triple Proposal} is also strategy-proof, because its output coincides with that of {\sc Double Proposal} (where arbitrariness in the algorithm is removed using some pre-specified order over agents).
However, in case neither (i) nor (ii) holds, {\sc Triple Proposal} is not strategy-proof anymore.
Actually, in such a setting, it seems difficult to attain approximation factors better than that of the arbitrary tie-breaking Gale--Shapley algorithm, while preserving strategy-proofness.
This is true at least for the marriage model because no strategy-proof algorithm can attain an approximation factor better than $2$, which easily follows from Hamada et al.~\cite{DBLP:conf/isaac/HamadaMY19}.
\end{remark}

\section{Maximum Gaps and Approximation Factors}\label{sec:approximation}
In this section, we analyze the approximation factors of {\sc Triple Proposal}, 
together with the maximum gaps for the four  models mentioned in Section~\ref{sec:intro}. 

For an instance $I$ of HRT-MSLQ, 
let $\opt(I)$ and $\wst(I)$ respectively denote the maximum and minimum scores over all 
stable matchings of $I$, and let $\alg(I)$  be the score of the output of {\sc Triple Proposal}. 
For a model $\cal I$ (i.e., subfamily of problem instances of HRT-MSLQ), let 
\[
\Lambda(\I)=\max_{I \in \I}\frac{\opt(I)}{\wst(I)} \ \  \mbox{ and } \ \ \app(\I)=\max_{I \in \I}\frac{\opt(I)}{\alg(I)}. 
\]

The maximum gap $\Lambda(\I)$ coincides with the worst approximation factor of 
a naive algorithm that first breaks ties arbitrarily and then applies the resident-oriented Gale--Shapley algorithm.
This equivalence is obtained by combining the following two facts: (1) for any stable matching $M$ of $I$, there exists an instance $I'$ with strict preferences such that $I'$ is obtained from $I$ via some tie-breaking and $M$ is a stable matching of $I'$, (2) for any instance $I'$ with strict preferences, the number of residents assigned to each hospital is invariant over all stable matchings. (For a more precise proof, see \cite[Proposition 20]{Goko}.)

In subsequent subsections, we tightly prove the values of the maximum gap $\Lambda(\I)$ and the approximation factor $\app(\I)$ of {\sc Triple Proposal}
for the general, uniform, and marriage models. The results for the resident-side master list model follow from those for the general model.

\begin{remark}\label{rem:double_proposal}
Let $\I$ be a model of HRT-MSLQ and $\I'$ be the subfamily of $\I$ consisting of all the instances in which preference lists of residents are strict.
The algorithm {\sc Double Proposal} applied to an instance $I'\in \I'$ has arbitrariness in the choice of a resident to be rejected at each rejection step because hospitals' lists can contain ties. 
Any choice in this arbitrariness can be viewed as a tie-breaking of hospitals' preference lists; more precisely, for any way of breaking the ties of hospitals' preference lists in $I'$, which results in an instance $I$ without ties, there is an execution of {\sc Double Proposal} for $I'$ in such a way that the output of {\sc Double Proposal} for $I'$ coincides with the output of the Gale-Shapley algorithm for $I$.
Therefore, any output of the arbitrary tie-breaking Gale--Shapley algorithm for an instance in $\I'$ can be an output of some execution of {\sc Double Proposal} for $\I'$.
This implies that the maximum gap for $\I'$ gives a lower bound on the approximation factor of {\sc Double Proposal} for $\I$.
The results in the second row of Table \ref{table1} are obtained in this manner.
\end{remark}

\subsection{Marriage Model}
We start with the model easiest to analyze.
Let ${\cal I}_{\rm Marriage}$ denote the family of instances of HRT-MSLQ, in which each hospital has an upper quota of $1$.
We call ${\cal I}_{\rm Marriage}$ the {\em marriage model}. 
By definition, $[\ell(h), u(h)]$ in this model is either $[0,1]$ or $[1,1]$ for each $h\in H$.
For this simple one-to-one matching case,
we can show the maximum gap and the approximation factor of {\sc Triple Proposal} using standard techniques. 

\begin{theorem}\label{thm:maxgap_mar}
The maximum gap for the marriage model satisfies $\Lambda(\I_{\rm Marriage})=2$.
This holds even if preference lists of residents are strict.
\end{theorem}
\begin{proof}
For any $I\in \I_{\rm Marriage}$, let $N$ and $M$ be stable matchings with $s(N)=\opt(I)$ and $s(M)=\wst(I)$.
Consider a bipartite graph $G=(R,H:N\cup M)$, where an edge used in both $N$ and $M$ is regarded as a length-two cycle in $G$.
Then, each connected component $C$ of $G$ is either an alternating cycle or an alternating path (possibly of length $0$).
Let $s_N(C)$ and $s_M(C)$, respectively, be the scores of $N$ and $M$ on the component $C$.
Then it suffices to show that $\frac{s_N(C)}{s_M(C)} \leq 2$ holds for any $C$.
If $C$ is an alternating cycle, then $s_N(h)=s_M(h)=1$ for every hospital $h$ on $C$, so $s_N(C)=s_M(C)$.
Suppose then that $C$ is an alternating path.
Note that a hospital $h$ on $C$ satisfies $s_N(h)=1$ and $s_M(h)=0$ only if $\ell(h)=1$ and $h$ is matched only in $N$ (and hence $h$ is an end terminal of a path).
At most one hospital in $C$ can satisfy this condition.
Note also that other hospitals $h$ on $C$ satisfy $s_M(h)=1$.
Now, let $k$ be the number of hospitals in $C$.
By the above argument, we have $\frac{s_N(C)}{s_M(C)} \leq \frac{k}{k-1}$, so if $k \geq 2$, we are done.
When $k=1$, there are three cases. 
If $C$ is a length-two path with a hospital $h$ being a midpoint, then $h$ is matched in both $N$ and $M$, so $s_N(C)=s_M(C)=1$.
If $C$ is an isolated vertex, then clearly $s_N(C)=s_M(C)=0$.
If $C$ is a length-one path $(r, h)$ of an edge of $N$ (resp. $M$), then $(r, h)$ blocks $M$ (resp. $N$), a contradiction.
Hence, $\Lambda(\I_{\rm Marriage}) \leq 2$.

We next provide an instance $I\in\I_{\rm Marriage}$ such that $\frac{\opt(I)}{\wst(I)} \geq  2$
and the preference lists of residents in $I$ are strict.
Let $I$ be  an instance defined as follows:
\begin{center}
\renewcommand\arraystretch{1.2}
\begin{tabular}{llllllllllllllllllllllllll}
$r_{1}$: & $h_{1}$ &  & & \hspace{15mm} & $h_{1}$ $[1, 1]$: &  (~$r_{1}$ &  $r_{2}$~) \\
$r_{2}$: & $h_{1}$ & $h_2$ & & \hspace{15mm} & $h_{2}$ $[1, 1]$: &   ~~$r_{2}$ &  &\\
\end{tabular}
\end{center}
Then, both $N=\{(r_1,h_1), (r_2,h_2)\}$ and $M=\{(r_2,h_1)\}$ are stable matchings and their scores are $s(N)=2$ and $s(M)=1$. 
\switch{\qed}{}\end{proof}

\begin{theorem}\label{thm:approx_mar}
The approximation factor of {\sc Triple Proposal} for the marriage model satisfies $\app(\I_{\rm Marriage})=1.5$.
\end{theorem}
\begin{proof}
For any $I\in \I_{\rm Marriage}$, let $M$ be the output of {\sc Triple Proposal} and $N$ be an optimal stable matching.
By arguments in the proof of Theorem~\ref{thm:maxgap_mar}, it suffices to show that there is no component of $G=(R,H; N\cup M)$ that forms an alternating path containing exactly two hospitals $h_1$ and $h_2$ with $s_N(h_1)=1$, $s_M(h_1)=0$, $s_N(h_2)=1$, and $s_M(h_2)=1$.
Suppose on the contrary that there is such a path. Then, $\ell(h_1)=1$, $h_1$ is matched only in $N$, and there exists a resident $r_1$ such that $N(r_1)=h_1$ and $M(r_1)=h_2$.
Since the path consists of only two hospitals, there are two cases: (i) $h_2$ is unmatched in $N$ and $\ell(h_2)=0$ and (ii) there exists a resident $r_2$ such that $N(r_2)=h_2$ and $M(r_2)=\varnothing$.
However, (ii) is impossible since otherwise $(r_2, h_2, r_1, h_1)$ is a length-3 $M$-augmenting path, which contradicts Lemma~\ref{lem:property2}.
Hence, assume (i). As $h_1$ is assigned no resident in $M$, we have $h_2=M(r_1)\succeq_{r_1} h_1$ by the stability of $M$.
Similarly, as $h_2$ is unmatched in $N$, the stability of $N$ implies $h_1=N(r_1)\succeq_{r_1} h_2$, and hence $h_1=_{r_1} h_2$.
Since $|M(h_2)|=1>\ell(h_2)$, Lemma~\ref{lem:property}(ii) implies $|M(h_1)|\geq \ell(h_1)=1$, 
which contradicts $|M(h_1)|=0$.
Thus we obtain $\frac{s(N)}{s(M)}\leq 1.5$ and hence $\app(\I_{\rm Marriage}) \leq 1.5$.

We next provide an instance $I\in \I_{\rm Marriage}$ such that $\frac{\opt(I)}{\alg(I)} \geq  1.5$.
Let $I$ be an instance defined as follows:
\begin{center}
\renewcommand\arraystretch{1.2}
\begin{tabular}{llllllllllllllllllllllllll}
$r_{1}$: & $h_{1}$ & $h_{2}$ & & \hspace{15mm} & $h_{1}$ $[1, 1]$: &  (~$r_{1}$ &  $r_{2}$~) \\
$r_{2}$: & $h_{1}$ & $h_{3}$ & & \hspace{15mm} & $h_{2}$ $[1, 1]$: &   ~~$r_{1}$ &  &\\
 & & & &  & $h_{3}$ $[0, 1]$: &  ~~$r_{2}$ &    \\
\end{tabular}
\end{center}
Execute {\sc Triple Proposal} for $I$ by prioritizing  $r_1$ over $r_2$ when an arbitrariness arises. 
That is, if $r_1$ and $r_2$ both satisfy the condition to be selected as $r'$ at Line~\ref{chosen2}, we choose $r_2$.  We then observe that the output of the algorithm is $M=\{(r_1, h_1), (r_2, h_3)\}$, which satisfies $s(M)=2$.
We see that the matching $N=\{(r_1, h_2), (r_2, h_1)\}$ stable and satisfies $s(N)=3$.
\switch{\qed}{}\end{proof}

\subsection{Uniform Model}
Let ${\cal I}_{\rm Uniform}$ denote the family of uniform problem instances of HRT-MSLQ, where an instance is called {\em uniform} if
if all hospitals have the same upper quotas and the same lower quotas.
In the rest of this subsection, we assume that $\ell$ and $u$ are nonnegative integers to represent 
the common lower and upper quotas, respectively, and let $\theta\coloneqq \frac{u}{\ell}~(\geq 1)$.
We call ${\cal I}_{\rm Uniform}$ the {\em uniform model}. 

Analysing this model is not so easy as the marriage model. Observe that, in the marriage model, our objective can be regarded as a weight maximization where the weight of each edge in $E$ is $1$ if it is incident to a hospital with quotas $[1,1]$ and $0$ otherwise. In the uniform model, however, such an interpretation is impossible and our objective function $\sum_{h\in H} \min\{ 1, \frac{|M(h)|}{\ell} \}$ is non-linear, which makes the analysis harder. We analyze this model by extending the techniques proposed in Goko et al.\cite{goko_et_al:LIPIcs.STACS.2022.31}.

We first provide the maximum gap.
The proof is presented in the Appendix.
\begin{theorem}\label{thm:maxgap_uni}
The maximum gap for the uniform model satisfies $\Lambda(\I_{\rm Uniform}) = \theta+1$.
This holds even if preference lists of residents are strict.
\end{theorem}
\begin{toappendix}
\begin{proof}[Proof of Theorem~\ref{thm:maxgap_uni}]
Here we only show that $\frac{\opt(I)}{\wst(I)}\leq \theta+1$ for any instance $I\in \I_{\rm Uniform}$. The other inequality is shown in Proposition~\ref{prop:uniform_gap_tight}.

Let $N$ and $M$ be stable matchings of maximum and minimum scores, respectively.
Define $v(N)=\ell \cdot s(N)$ and $v(M)=\ell\cdot s(M)$. Then, $\frac{v(N)}{v(M)}=\frac{s(N)}{s(M)}=\frac{\opt(I)}{\wst(I)}$. 
For residents assigned to each $h\in H$, we set $p_{M}(r)=1$ for $\min \{ \ell(h), |M(h)| \}$ residents $r$, and $p_{M}(r)=0$ for the remaining $|M(h)|-\min \{ \ell(h), |M(h)| \}$ residents $r$. Then, we have $v(M)=\sum_{r\in R}p_M(r)$.

Let $H^{*}$ be the set of hospitals that are full in $M$.
Partition the residents as $R=R_{1} \cup R_{2} \cup R_{3}$, where $R_{1} = \set{ r|M(r) \in H^{*} }$, $R_{2} = \set{ r | M(r) \in H \setminus H^{*} }$, and $R_{3} = \set{ r |M(r)=\varnothing}$.
Let $x$ and $y$ be the numbers of residents $r$ in $R_{2}$ such that $p_{M}(r)=1$ and $p_{M}(r)=0$, respectively.
Then the total number of residents is $n=|H^{*}| u+x+y+|R_{3}|$.
By definition, we have that $v(M)=|H^{*}| \ell+ x$.
We now bound the value of $v(N)$.
Note that the list of any $r\in R_{3}$ contains only hospitals in $H^{*}$, since otherwise a hospital in $H\setminus H^*$ appearing on the list forms a blocking pair with $r$. Hence, at most $|H^{*}| \ell$ residents $r$ in $R_{3}$ can have $p_{N}(r)=1$.
Then, it follows that $v(N) \leq |H^{*}| u+x+y+|H^{*}| \ell$.

If $M(h)$ includes a resident $r$ such that $p_{M}(r)=0$, then $M(h)$ includes $\ell$ residents $r'$ such that $p_{M}(r')=1$.
Thus, we have $\frac{y}{x} \leq \frac{u-\ell}{\ell} = \theta -1$, which implies $\frac{x+y}{x} \leq \theta$.
Therefore
\begin{align*}
\frac{v(N)}{v(M)} &\leq \frac{|H^{*}|(u+\ell) +(x+y)}{|H^{*}|\ell + x} \\
&\leq \frac{|H^{*}|\ell(1+\theta) +x\theta}{|H^{*}|\ell + x}\\
&= \frac{|H^{*}|\ell(1+\theta) +x(1+\theta)-x}{|H^{*}|\ell + x}
\leq 1+\theta.
\end{align*}

\switch{\qed}{}\end{proof}

\begin{proposition}\label{prop:uniform_gap_tight}
There is an instance $I\in \I_{\rm Uniform}$ such that $\frac{\opt(I)}{\wst(I)} \geq  \theta+1$ holds.
This holds even if preference lists of residents are strict.
\end{proposition}

\begin{proof}
Consider an instance $I$ with residents $\{r_1, r_2, \dots, r_{u}\} \cup \{r'_1, r'_2, \dots, r'_{u} \}$ and hospitals $\{h_1, h_2, \dots, h_{u}, x\}$. 
Preference lists and quotas are as follows, where ``(~~~$R$~~~)'' means a tie consisting of all residents.
\begin{center}
\renewcommand\arraystretch{1.2}
\begin{tabular}{llllllllllllllllll}
$r_{i}$: & $x$ & $h_{i}$ & &\hspace{15mm} & $h_i$ $[\ell, u]$: &  $r_{i}$ \hspace{15mm}\\
$r'_{i}$: & $x$ & & &\hspace{15mm} & $x$ $[\ell, u]$: &  (~~~$R$~~~)\\
\end{tabular}
\end{center}
Note that ties appear only in hospitals' preference lists.

Define matchings $N=\set{(r_i, h_i) \mid i=1, 2, \dots, u} \cup \set{(r'_i, x) \mid i=1, 2, \dots, u}$ and $M=\set{(r_i, x) \mid i=1, 2, \dots, u}$.
Then, $s(N)=1+\frac{u}{\ell} = \theta+1$ while $s(M)=1$.
Thus we obtain $\frac{\opt(I)}{\wst(I)} \geq \theta+1$.
\switch{\qed}{}\end{proof}
\end{toappendix}

We next show the approximation factor of our algorithm. 
\begin{theorem}\label{thm:IL-uniform-alg}
The approximation factor of {\sc Triple Proposal} for the uniform  model satisfies $\app(\I_{\rm uniform})=\frac{\theta}{2}+1$.
\end{theorem}
We present a full proof of Theorem~\ref{thm:IL-uniform-alg} in the Appendix. 
The framework of the proof is similar to that of Theorem~9 of Goko et al.~\cite{Goko}. 
Here we explain an outline of the proof of $\app(\I_{\rm uniform})\leq \frac{\theta}{2}+1$ putting emphasis on the difference from~\cite{Goko}.

Let $M$ be the output of the algorithm and $N$ be an optimal stable matching. Consider a bipartite (multi-)graph $(R, H; M\cup N)$.
To complete the proof, it suffices to show that 
the approximation factor is attained in each connected component. Take any connected component and let $H_0$ be the set of all hospitals $h$ with $s_N(h)>s_M(h)$ in the component. That is, hospitals in $H_0$ get larger scores in $N$ than in $M$, which implies that those hospitals are deficient in $M$. We then categorize hospitals and residents using a breadth-ﬁrst search as follows (see Figure~\ref{fig:graph2}):
consider $NM$-alternating paths that start at $H_0$ with edges in $N$. Let $R_0$ and $R_1$ be the sets of residents reachable from $H_0$ via those paths of length $1$ and $3$, respectively and let $H_1$ be the set of hospitals reachable from $H_0$ via those paths of length $2$. The set of remaining residents (resp., hospitals) in the component is denoted as $R_2$ (resp., $H_2$).
\begin{figure}[t]
	\begin{center}
		\includegraphics[width=0.37\hsize]{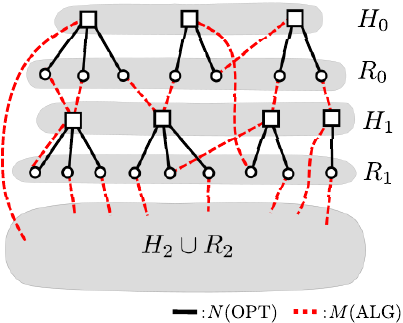}
	\end{center}
	\caption{\small An example of a connected component in $N\cup M$ for the case $[\ell, u]=[2,3]$.
	Hospitals and residents are represented by squares and circles, respectively.
	The matchings $N$ and $M$ are represented by solid (black) lines and dashed (red) lines, 
	respectively. }
	\label{fig:graph2}
\end{figure}

In Goko et al.~\cite{Goko}, the authors upper bound the ratio of the score of $N$ to that of $M$ by a kind of counting argument on the number of residents in the component: they estimate this number in two ways using the stability of $N$ and $M$ and the property in Lemma~\ref{lem:property}. In the estimations, they use the fact that all residents in the component are matched in $M$, which is not true for the current incomplete list setting. However, thanks to the nonexistence of length-3 $M$-augmenting paths shown in Lemma~\ref{lem:property2}, we can guarantee that all residents in $R_0 \cup R_1$ are matched in $M$ (residents in $R_0$ are matched by the stability of $M$). Utilizing this fact, we can obtain the required upper bound.

\begin{toappendix}
\begin{proof}[Proof of Theorem~\ref{thm:IL-uniform-alg}]
Here we only show $\app(\I_{\rm Uniform})\leq \frac{\theta}{2}+1$,
since this together with Proposition~\ref{prop:uniform-approx-tight} shown later
implies the required equality.

Let $M$ be the output of the algorithm and let $N$ be an optimal stable matching.
Consider a bipartite graph $(R, H; M\cup N)$, which may have multiple edges.
To complete the proof, it is sufficient to show that 
the approximation factor is attained in each component of the graph.
Take any connected component and let $R^*$ and $H^*$ respectively denote the set of residents and hospitals in the component.
We define a partition $\{H_0, H_1,H_2\}$ of $H^*$ and
a partition $\{R_0, R_1,R_2\}$ of $R^*$ as follows (See Figure~\ref{fig:graph2-second}, which is taken from \cite{goko_et_al:LIPIcs.STACS.2022.31}). 
First, we set
\begin{align*}
&H_0\coloneqq \set{h\in H^*|s_N(h)>s_M(h)} \mbox{ and}\\
&R_{0}\coloneqq \set{r\in R^*|N(r)\in H_0}.
\end{align*}
That is, $H_0$ is the set of all hospitals in the component for which the optimal stable matching
$N$ gets scores larger than $M$. The set $R_0$ consists of residents assigned to $H_0$ in the optimal matching $N$.
We then define
\begin{align*}
&H_{1}\coloneqq \set{h\in H^*\setminus H_0|\exists r\in R_{0}:M(r)=h},\\
&R_{1}\coloneqq \set{r\in R^*|N(r)\in H_1},\\
&H_2\coloneqq H^*\setminus(H_0\cup H_1), \text{~~and~~}\\  
&R_2\coloneqq R^*\setminus(R_0\cup R_1).
\end{align*}
\begin{figure}[t]
	\begin{center}
		\includegraphics[width=0.37\hsize]{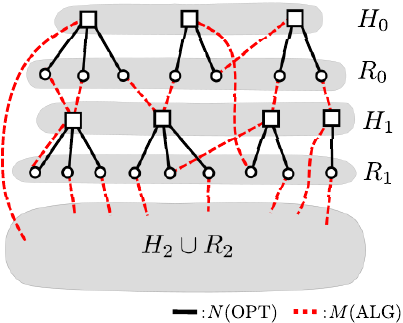}
	\end{center}
	\caption{\small An example of a connected component in $N\cup M$ for the case $[\ell, u]=[2,3]$.
	Hospitals and residents are represented by squares and circles, respectively.
	The matchings $N$ and $M$ are represented by solid (black) lines and dashed (red) lines, 
	respectively. }
	\label{fig:graph2-second}
\end{figure}
For convenience, we use a scaled score $v_M(h)\coloneqq \ell\cdot s_M(h)=\min\{ \ell, |M(h)|\}$ for each $h\in H$ and write $v_M(H')\coloneqq \sum_{h \in H'} v_M(h)$ for any $H'\subseteq H$. 
We define $v_N\coloneqq \ell\cdot s_N(h)=\min\{ \ell, |N(h)|\}$ similarly.
We now show the following inequality, which completes the proof:
\begin{equation}
\frac{v_N(H_0\cup H_1\cup H_2)}{v_M(H_0\cup H_1\cup H_2)}\leq \frac{\theta}{2}+1.
\label{eq:goal}
\end{equation}
Suppose $v_N(H^*)>v_M(H^*)$ since otherwise the claim is trivial. By the definition of $H_0$, we have
\begin{equation}
v_N(H_0)>v_M(H_0), \quad v_N(H_1)\leq v_M(H_1),\quad v_N(H_2)\leq v_M(H_2).
\label{eq:ineqs}
\end{equation}

Let $\alpha\coloneqq v_N(H_0\cup H_1)-v_M(H_0\cup H_1)$, which is positive by
$v_N(H_0\cup H_1)-v_M(H_0\cup H_1)=v_N(H^*)-v_M(H^*)+\{v_M(H_2)-v_N(H_2)\}>0$.
In addition, by $v_N(H_1)\leq v_M(H_1)$, we obtain  $\alpha\leq v_N(H_0)+\{v_N(H_1)-v_M(H_1)\}\leq  v_N(H_0)=\sum_{h \in H_0} \min\{\ell, |N(h)|\}\leq \sum_{h \in H_0} |N(h)|=|R_0|$.
Because $M$ assigns each resident $r\in R_0$ to some hospital in $H_0\cup H_1$ by the definition of $H_1$,
we have $\sum_{h \in H_0\cup H_1} |M(h)|\geq |R_{0}|$.
Combining these implies $\sum_{h \in H_0\cup H_1} |M(h)|\geq\alpha$. 
Since $v_M(h)=\min\{\ell, |M(h)|\}\geq \frac{1}{\theta}|M(h)| $ for any $h\in H_0\cup H_1$, we obtain
\begin{equation}
v_M(H_0\cup H_1)\geq \frac{\alpha}{\theta}.
\label{eq:alpha}
\end{equation}
Let $\beta\coloneqq v_N(H_2)$. Then, we have
\begin{equation}
	v_N(H_0\cup H_1\cup H_2)=\alpha+\beta+v_M(H_0\cup H_1).
	\label{eq:beta}
\end{equation}
We separately consider two cases: (i) $\beta\geq \frac{\alpha}{\theta}$ and (ii) $\beta\leq \frac{\alpha}{\theta}$.

First, consider the case (i).
By \eqref{eq:ineqs}, we have 
$v_M(H_0\cup H_1\cup H_2)\geq v_M(H_0\cup H_1)+v_N(H_2)=\beta+v_M(H_0\cup H_1)$.
Combining this with the equation \eqref{eq:beta}, we obtain \eqref{eq:goal} in this case as follows. 
\begin{eqnarray*}
	\frac{v_N(H_0\cup H_1\cup H_2)}{v_M(H_0\cup H_1\cup H_2)} &\leq & 
	\frac{\alpha+\beta+v_M(H_0 \cup H_1)}{\beta+v_M(H_0 \cup H_1)} \\
	& = & 1+\frac{\alpha}{\beta+v_M(H_0 \cup H_1)} \\
	& \leq & 1 + \frac{\alpha}{\frac{\alpha}{\theta}+\frac{\alpha}{\theta}} \\
	& = & 1+\frac{\theta}{2}.
\end{eqnarray*}
Here the second inequality follows from the inequality \eqref{eq:alpha} and the condition (i). 

We next consider the case (ii) $\beta\leq \frac{\alpha}{\theta}$.
We prepare the following claim.
\begin{\switch{myclaim}{claim}}\label{claim:R0R1}
Any $r\in R_0$ is assigned to $H_0\cup H_1$ in $M$ and satisfies $M(r)\succeq _r N(r)$.
Any $r\in R_1$ is assigned to some hospital in $M$.
\end{\switch{myclaim}{claim}}
\begin{proof}
By the definition of $H_0$, any hospital is deficient in $M$.
Also, by the definition of $R_0$, any $r\in R_0$ is assigned to some $h\in H_0$ in $N$.
As $h$ is deficient in $M$, the stability of $M$ implies that $r$ is matched in $M$ and $M(r)\succeq _r N(r)$.
In addition, the definitions of $H_1$ implies $M(r)\in H_0\cup H_1$.
As for a resident $r'\in R_1$, the definitions of $R_0$ and $H_1$ imply that
there exist edges $(r,h),(r',h')\in N\setminus M$, $(r,h')\in M$ with $h'\in H_1$, $r\in R_0$, and $h\in H_0$.
Thus, Lemma~\ref{lem:property2} implies that $r'$ is matched in $M$. 
\switch{\qed}{}\end{proof}
The first statement of Claim~\ref{claim:R0R1} implies that, for any $h\in H_1$, 
all residents $r$ in $M(h)\cap R_0$ satisfy $M(r)\succeq _r N(r)$.
Therefore, the following $\{H_1^{\succ}, H_1^{=}\}$ defines a bipartition of $H_1$:
\begin{align*}
&H_1^{\succ}\coloneqq \set{h\in H_1|\exists r\in M(h)\cap R_0: h\succ_r N(r)},\\
&H_1^{=}\coloneqq \set{h\in H_1|\forall r\in M(h)\cap R_0: h=_r N(r)}.
\end{align*}
We now intend to show the following inequality by estimating $|R_0|+|R_1|$ in two ways.
\begin{equation}
v_M(H_2) \geq \frac{\alpha}{\theta},  \label{eq:R-size3}
\end{equation}
For the first estimation, we further partition the set $R_1$  into 
$R_1^{\succ}\coloneqq \set{r\in R_1|N(r)\in H_1^{\succ}}$ and  $R_1^{=}\coloneqq \set{r\in R_1| N(r)\in H_1^{=}}$.
By the stability of $N$, 
 each $h\in H_1^{\succ}$ is full in $N$, 
since there exists a resident $r\in R_0$ with $h\succ_r N(r)$, implying that  $|N(h)|=u$ and  $v_N(h)=\ell$. 
Thus we have 
$|R_1^{\succ}|=u\cdot |H_1^{\succ}| = \frac{u}{\ell}\cdot v_N(H_1^{\succ})=\theta \cdot v_N(H_1^{\succ})$.
Additionally, since each $h\in H_1^{\succ}$ satisfies $v_M(h)\geq v_N(h)$ by $h\not\in H_0$,  we have $v_M(h)=v_N(h)=\ell$, which implies  $v_M(H_1^{\succ})=v_N(H_1^{\succ})$. 
Furthermore, by definition,  $|R_0\cup R_1^{=}|\geq v_N(H_0\cup H_1^{=})$.    
Combining them, we obtain 
\begin{equation}
|R_0|+|R_1|\geq v_N(H_0\cup H_1^{=})+\theta\cdot v_N(H_1^{\succ})=v_N(H_0\cup H_1)+(\theta-1)\cdot v_M(H_1^{\succ}).\label{eq:R-size}
\end{equation}

For the second estimation of $|R_0|+|R_1|$, we define another partition $\{S_1, S_2\}$ of $R_0\cup R_1$ depending on the matching $M$: 
\begin{align*}
&S_{1}\coloneqq \set{r\in R_0\cup R_1|M(r)\in H_0\cup H_1^{=}},\\
&S_{2}\coloneqq (R_0\cup R_1)\setminus S_1.
\end{align*}
We show that $|S_1|\leq v_M(H_0\cup H_1^{=})$  and $|S_{2}|\leq \theta \cdot v_M(H_1^{\succ}\cup H_2)$. 
Since any $h\in H_0$ satisfies $\ell\geq v_N(h)>v_M(h)$, we have  $v_M(h)=|M(h)|$.
For each $h\in H_1^{=}$, there exists $r\in R_0$ with $M(r)=h$ and $M(r)=_r N(r)$.
Since for any $r\in R_0$, the hospital $h'\coloneqq N(r)$ belongs to $H_0$,  we have $|M(h')|<\ell$. 
From this together with  Lemma~\ref{lem:property}, we have $|M(h)|\leq \ell$, which shows that  $v_M(h)=|M(h)|$ for each $h\in H_1^{=}$.
Thus, we obtain $|S_1|\leq \sum_{h\in H_0\cup H_1^{=}} |M(h)|=v_M(H_0\cup H_1^{=})$.
The equality $|S_{2}|\leq \theta \cdot v_M(H_1^{\succ}\cup H_2)$ follows from the fact that all residents in $S_2$ are assigned to
$H_1^{\succ}\cup H_2$, because Claim~\ref{claim:R0R1} implies that all residents in $R_0\cup R_1$ are matched in $M$. 

By these equalities,  we have 
\begin{equation*}
|R_0|+|R_1| = |S_1| + |S_2| \leq  v_M(H_0\cup H_1) + (\theta-1) \cdot v_M(H_1^{\succ})+\theta \cdot v_M(H_2),
\end{equation*}
which together with \eqref{eq:R-size} provides the inequality \eqref{eq:R-size3}. 

By using \eqref{eq:beta} and (\ref{eq:R-size3}), we obtain the required inequality \eqref{eq:goal} also for the case (ii) as follows:
\begin{eqnarray*}
\frac{v_N(H_0\cup H_1\cup H_2)}{v_M(H_0\cup H_1\cup H_2)} &=& 
\frac{\alpha+\beta+v_M(H_0 \cup H_1)}{v_M(H_2)+v_M(H_0 \cup H_1)} \\
& \leq & \frac{\alpha+\frac{\alpha}{\theta}+v_M(H_0 \cup H_1)}{\frac{\alpha}{\theta}+v_M(H_0 \cup H_1)} \\
&=&1+\frac{\alpha}{\frac{\alpha}{\theta}+v_M(H_0 \cup H_1)}\\
&\leq &1+\frac{\alpha}{\frac{\alpha}{\theta}+\frac{\alpha}{\theta}}\\
&=&1+\frac{\theta}{2}.\\
\end{eqnarray*}
Here, the first inequality follows from the condition (ii) and \eqref{eq:R-size3},
and the second inequality follows from \eqref{eq:alpha}.
\switch{\qed}{}\end{proof}

\begin{proposition}\label{prop:uniform-approx-tight}
There is an instance $I\in \I_{\rm Uniform}$ such that $\frac{\opt(I)}{\alg(I)} \geq  \frac{\theta}{2}+1$ holds for any values of $u$ and $\ell$.
\end{proposition}
\begin{proof}
Consider an instance $I$ that consists of $3u$ residents $\{a_1, a_2,\dots,a_u\}\cup \{b_1, b_2,\dots,b_u\}\cup \{c_1, c_2,\dots,c_u\}$ and $u+2$ hospitals
$\{h_1,h_2,\dots,h_u\}\cup \{x\}\cup \{y\}$.
The preference lists are given as follows.
\begin{center}
\renewcommand\arraystretch{1.2}
\begin{tabular}{llllllllllllllllll}
$a_{i}$: & $x$ & $h_i$   &\hspace{15mm} & $h_i$ $[\ell, u]$: &  $a_i$ \hspace{15mm}\\
$b_{i}$: & $x$ & $y$  &\hspace{15mm} & $x$ $[\ell, u]$: &  (~$a_{1}~a_{2}\cdots a_u~b_1~b_2\cdots b_{u}$~) \\
$c_{i}$: & $y$  &  &\hspace{15mm}& $y$ $[\ell, u]$: &  $b_{1}~b_{2}\cdots b_u~c_1~c_2\cdots c_{u}$  \\
\end{tabular}
\end{center}
Execute {\sc Triple Proposal} for $I$ by prioritizing residents in $A$ over those in $B$ whenever an arbitrariness arises. That is, if a resident in $A$ and that in $B$ both satisfy the condition to be selected as $r'$ at Line~\ref{chosen2}, we choose a resident in $B$.
We then observe that the output of the algorithm is
$M=\set{(a_{i},x)|1\leq i\leq u}\cup \set{(b_{i},y)|1\leq i\leq u}$, and hence $s(M)=2$.
On the other hand, a matching  
$N=\set{(a_{i},h_i)|1\leq i\leq u}\cup \set{(b_{i},x)|1\leq i\leq u}\cup \set{(c_{i},y)|1\leq i\leq u}$ is stable and satisfies $s(N)=\frac{u}{\ell}+2=\theta+2$.
Therefore, $\frac{\opt(I)}{\alg(I)} \geq \frac{s(N)}{s(M)}\geq  \frac{\theta}{2}+1$.
\switch{\qed}{}\end{proof}
\end{toappendix}

\subsection{General Model and $R$-side ML Model}\label{subsec:general}
Let ${\cal I}_{\rm Gen}$ denote the family of all instances of HRT-MSLQ, which we call the {\em general model}. 
In addition, let $\I_{\rm R\mathchar`-ML}$ denote the subfamily of ${\cal I}_{\rm Gen}$ representing the {\em resident-side master list model}.
That is, in an instance in $\I_{\rm R\mathchar`-ML}$, there exists a master preference list from which each resident's list is obtained by deleting unacceptable hospitals. 
In this section, we mainly work on the general model and obtain results on the resident-side master list model as  consequences.

As in the case of the uniform model, the non-linearlity of the objective function, i.e., the total satisfaction ratio $\sum_{h\in H} \min\{ 1, \frac{|M(h)|}{\ell(h)} \}$ prevents us from applying the standard techniques used for the marriage model. Furthermore, due to the non-uniformness of the quotas, we cannot apply the counting argument used for the uniform model. 
It appears, however, that the techniques used in Goko et al.~\cite{Goko} can apply even for the current incomplete list setting.

We first provide two theorems on the maximum gap.
\begin{theorem}\label{thm:maxgap_gen}
The maximum gaps for the general and resident-side master list models are $n+1$, i.e., $\Lambda(\I_{\rm Gen})=\Lambda(\I_{\rm R\mathchar`-ML})=n+1$.
\end{theorem}
\begin{toappendix}
\begin{proof}[Proof of Theorem~\ref{thm:maxgap_gen}]
We first show $\frac{\opt(I)}{\wst(I)}\leq n+1$ for any $I\in \I_{\rm Gen}$.
Let $N$ and $M$ be stable matchings with $s(N)=\opt(I)$ and $s(M)=\wst(I)$, respectively.
Recall that $\ell(h)\leq n$ is assumed for any hospital $h$.
Let $H_0\subseteq H$ be the set of hospitals $h$ with $\ell(h)=0$.
Then
\begin{align*}
&\textstyle s(N)=|H_0|+\sum_{h\in H\setminus H_0}\min\{1,\frac{|N(h)|}{\ell(h)}\}\leq\switch{}{
|H_0|+\sum_{h\in H\setminus H_0}\min\{1,\frac{|N(h)|}{1}\}\leq} |H_0|+n,\\
&\textstyle s(M)= 
|H_0|+\sum_{h\in H\setminus H_0}\min\{1,\frac{|M(h)|}{\ell(h)}\}\geq |H_0|+\sum_{h\in H\setminus H_0}\min\{1,\frac{|M(h)|}{n}\}.
\end{align*}
In case $|H_0|=0$, we have
$\sum_{h\in H\setminus H_0}\min\{1,\frac{|M(h)|}{n}\}=\sum_{h\in H}\min\{1,\frac{|M(h)|}{n}\}\geq 1$,
and hence $\frac{s(N)}{s(M)}\leq \frac{n}{1}=n$.
In case $|H_0|\geq 1$, we have $s(M)\geq|H_0|$, and $\frac{s(N)}{s(M)}\leq \frac{|H_0|+n}{|H_0|}=1+\frac{n}{|H_0|}\leq 1+n$.
Thus, $\frac{\opt(I)}{\wst(I)}\leq n+1$ for any instance $I$.

We next show that there is an instance $I\in \I_{\rm R\mathchar`-ML}$ with $\frac{\opt(I)}{\wst(I)}\geq n+1$.
Let $I$ be an instance consisting of $n$ residents $\{r_1,r_2,\dots,r_n\}$ and
$n+1$ hospitals $\{h_1,h_2,\dots,h_{n+1}\}$ 
such that 
\begin{itemize}
\item  \switch{every resident's preference list}{the preference list of every resident} consists of a single tie containing all hospitals,
\item \switch{every hospital's preference list}{the preference list of every hospital} is an arbitrary complete list without ties, and
\item $[\ell(h_i), u(h_i)]=[1,1]$ for $i=1,2,\dots,n$ and $[\ell(h_{n+1}), u(h_{n+1})]=[0, n]$.
\end{itemize}
Since any resident is indifferent among all hospitals, a matching
is stable whenever all residents are assigned.
Let $N=\set{(r_i,h_i)|i=1,2,\dots,n}$ and $M=\set{(r_i,h_{n+1})|i=1,2,\dots,n}$.
Then, $s(N)=n+1$ while $s(M)=1$. Thus we obtain $\frac{\opt(I)}{\wst(I)}\geq n+1$.
\switch{\qed}{}\end{proof}
\end{toappendix}
\begin{theorem}\label{thm:maxgap_gen_H}
If the preference lists of residents are strict,
the maximum gaps for the general and resident-side master list models are $\lfloor \frac{n}{2} \rfloor+1$. 
\end{theorem}
\begin{toappendix}
\begin{proof}[Proof of Theorem~\ref{thm:maxgap_gen_H}]
Here we only show that $\frac{\opt(I)}{\wst(I)}\leq \lfloor \frac{n}{2} \rfloor+1$ holds for any instance $I\in \I_{\rm Gen}$ in which residents have strict preferences. The other inequality is shown in Proposition~\ref{prop:below} with instances in $\I_{\rm R\mathchar`-ML}$.

In case $n=1$, a stable matching is unique, which assigns the unique resident to her best choice hospital.
Hence the claim clearly holds.

We will assume $n\geq 2$.
Let $H_{0}$ be the set of hospitals $h$ such that $\ell(h)=0$, and $H_{f}$ be the set of hospitals $h$ that is the first choice of at least $u(h)$ residents.

\begin{\switch{myclaim}{claim}}\label{claim:Hf}
For any instance $I$, the following holds: 
(1) In any stable matching, every hospital in $H_{f}$ is full.
(2) $\opt(I) \leq n+|H_{0} \setminus H_{f}|$.
(3) $\wst(I) \geq |H_{0} \cup H_{f}|$.
\end{\switch{myclaim}{claim}}

\begin{proof}
(1) If $h \in H_{f}$ is not full in a matching $M$, there is a resident $r \not\in M(h)$ whose first choice is $h$.
Then $(r, h)$ blocks $M$ so $M$ is unstable.
(2) Consider an optimal stable matching.
A nonempty hospital can gain a score at most 1 and there are at most $n$ such hospitals. 
An empty hospital can gain a score of 1 only if it is in $H_{0}$, but by (1) of this proposition, hospitals in $H_{f}$ are nonempty and are excluded.
(3) Consider an arbitrary stable matching $M$.
If $h \in H_{0}$, $s_{M}(h)=1$ by definition.
If $h \in H_{f}$, $s_{M}(h)=1$ since $h$ is full by (1) of this proposition.
\switch{\qed}{}\end{proof}
We assume $|H_f|\geq 1$ since otherwise all residents are assigned to their first choices in any stable matching, which implies $\wst(I)=\opt(I)$.
We consider the following three cases: $\wst(I)\geq 3$, $2\leq \wst(I)<3$, and  $\wst(I)<2$.

In case $\wst(I)\geq 3$, by Claims~\ref{claim:Hf}(2) and \ref{claim:Hf}(3), we have
$\frac{\opt(I)}{\wst(I)} \leq \frac{n+|H_{0}|}{\wst(I)} \leq \frac{n}{\wst(I)} + \frac{|H_{0}|}{\wst(I)} \leq \frac{n}{3} + 1 \leq 1+\lfloor \frac{n}{2} \rfloor$.

In case $2\leq \wst(I)<3$, Claim~\ref{claim:Hf}(3) and $|H_f|\geq 1$ imply $|H_0\setminus H_f|\leq 1$. 
Then, Claim~\ref{claim:Hf}(2) implies $\frac{\opt(I)}{\wst(I)} \leq \frac{n+1}{2} \leq 1+\lfloor \frac{n}{2} \rfloor$.

Finally, we consider the case $\wst(I) < 2$. Since $1\leq |H_f|\leq \wst(I)<2$, $H_f$ consists of a single hospital.
Let $H_{f}=\{ h^{*} \}$.
Also, let $N$ and $M$ be stable matchings of maximum and minimum scores, respectively.
Partition the residents as $R=R_{1} \cup R_{2} \cup R_{3}$ where $R_{1} = \set{ r\in R | N(r)= h^{*}}$, $R_{2} = \set{ r\in R | N(r)\in H\setminus \{h^*\}}$, and $R_{3} = \set{ r\in R| N(r)=\varnothing}$, and define $S = \set{ r\in R| M(r)= h^{*}}$.
By Claim \ref{claim:Hf}(1), $h^{*}$ is full in both $N$ and $M$, so $|R_{1}|=|S|$.

For each resident $r \in R_{2} \setminus S$, let $b(r)$ be $r$'s first choice hospital excluding $h^{*}$.
Note that $b(r)$ exists for any $r$ since residents in $R_{2}$ are matched to a hospital other than $h^{*}$ in $M$.
We show that $M(r)=b(r)$ for any $r \in R_{2} \setminus S$.
Suppose not.
Then $b(r) \succ_{r} M(r)$ or $M(r)=\varnothing$, so for $M$ to be stable, $b(r)$ must be full in $M$.
But then $s_{M}(b(r))=1$, which, together with $s_{M}(h^{*})=1$, implies $\wst(I) \geq 2$, a contradiction.

Let $H'=\set{h\in H|M(h)\cap(R_{2} \setminus S)\neq \emptyset}$. Note that $h^*\not\in H'$ by definition, and hence any $h\in H'$ satisfies $s_M(h)<1$.
We claim that any $h\in H'$ satisfies $|M(h)\cap(R_{2} \setminus S)|\leq |N(h)|$. 
Suppose not. Then, $|N(h)|$ is undersubscribed and there is a resident $r\in R_{2} \setminus S$ such that $M(r)=h\neq N(r)$. As we have $M(r)=b(r)$, this $r$ satisfies $h \succ_{r} N(r)$ or $N(r)=\varnothing$ while $h$ is undersubscribed. Then, $(r, h)$ blocks $N$, a contradiction.

For each $h\in H'$, consider any injection from $M(h)\cap (R_{2} \setminus S)$ to $N(h)$.
Let $T$ be the union of the ranges of such injections over all $h\in H'$.
The direct sum of the injections defines a bijection between $R_2\setminus S$ and $T$.
Since any $h\in H'$ satisfies $s_M(h)<1$, we have $s_M(h)=\frac{|M(h)|}{\ell(h)}$, i.e., each resident $r\in M(h)\cap(R_2\setminus S)$ contributes $\frac{1}{\ell(h)}$ to the score $s(M)$. 
Let $\alpha$ be the total contribution of residents in $R_2\setminus S$ to $s(M)$.
Then, we have $s(M)\geq s_M(h^*)+\alpha=1+\alpha$.
Since $T \subseteq R_{2}$, the bijection between $R_2\setminus S$ and $T$ implies that the total contribution of residents in $R_2$ to $s(N)$ is at most $\alpha+|R_2\setminus T|$. Then, $s(N)\leq s_N(h^*)+\alpha+|R_2\setminus T|=1+\alpha+|R_2\setminus T|$.

Note that $T \subseteq R_{2}$ and $|T|=|R_{2} \setminus S|$ imply $|R_{2} \setminus T|=|R_{2} \cap S|$.
Also, by $|R_{2} \cap S| \leq |S| = |R_{1}|$ and $|R_{1}|+|R_{2} \cap S| \leq n$, we have $|R_{2} \cap S| \leq \lfloor \frac{n}{2} \rfloor$. Then, $|R_{2} \setminus T|\leq \lfloor \frac{n}{2} \rfloor$ follows, and we have
\[
\frac{\opt(I)}{\wst(I)} \leq \frac{1+\alpha+|R_2\setminus T|}{1+\alpha} \leq 1+\frac{\lfloor \frac{n}{2} \rfloor}{1+\alpha}  \leq 1+\left\lfloor \frac{n}{2} \right\rfloor.
\]
\switch{\qed}{}\end{proof}

\begin{proposition}\label{prop:below}
For any natural number $n$, there is an instance $I \in \I_{\rm R\mathchar`-ML}$ with $n$ residents such that preference lists of residents are strict and $\frac{\opt(I)}{\wst(I)} \geq 1+\lfloor \frac{n}{2} \rfloor$.
\end{proposition}
\begin{proof}
Consider an instance $I$ with residents $\{r_1, r_2, \dots, r_{\lfloor \frac{n}{2} \rfloor}\} \cup \{r'_1, r'_2, \dots, r'_{\lceil \frac{n}{2} \rceil}\}$ and hospitals $\{h_1, h_2, \dots, h_{\lfloor \frac{n}{2} \rfloor}, x\}$. 
Preference lists and quotas are as follows, where ``(~~~$R$~~~)'' means a tie consisting of all residents.

\begin{center}
\renewcommand\arraystretch{1.2}
\begin{tabular}{llllllllllllllllll}
$r_{i}$: & $x$ & $h_{i}$ & &\hspace{15mm} & $h_i$ $[1, 1]$: &  $r_{i}$ \hspace{15mm}\\
$r'_{i}$: & $x$ & & &\hspace{15mm} & $x$ $[\lceil\frac{n}{2}\rceil, \lceil\frac{n}{2}\rceil]$: &  (~~~$R$~~~)\\
\end{tabular}
\end{center}

This instance satisfies the conditions in the statement, where any strict order on $H$ with $x$ at the top can be a master list.

Let $N=\set{(r_i, h_i) \mid i=1, 2, \dots, \lfloor \frac{n}{2} \rfloor} \cup \set{(r'_i, x) \mid i=1, 2, \dots, \lceil \frac{n}{2} \rceil}$.
Define $M=\set{(r_i, x) \mid i=1, 2, \dots, \lfloor \frac{n}{2} \rfloor}$ and let $\tilde{M}=M$ if $n$ is even and $\tilde{M}=M \cup \set{(r'_{1}, x)}$ if $n$ is odd.
Then, $s(N)=1+\lfloor \frac{n}{2} \rfloor$ while $s(\tilde{M})=1$. Thus we obtain $\frac{\opt(I)}{\wst(I)}\geq 1+\lfloor \frac{n}{2} \rfloor$.
\switch{\qed}{}\end{proof}
\end{toappendix}

We next show that the approximation factor of our algorithm is $\phi(n)$,
where $\phi$ is a function of $n=|R|$ defined by
\begin{equation*}\label{eq:phi}
\phi(n)=
\begin{cases}
1&n=1,\\
\frac{3}{2}&n=2,\\
\frac{n(1+\lfloor\frac{n}{2}\rfloor)}{n+\lfloor \frac{n}{2}\rfloor}&n\geq 3.
\end{cases}
\end{equation*}
\begin{theorem}\label{thm:general-approximable}
The approximation factor of {\sc Triple Proposal} for the general model satisfies $\app(\I_{\rm Gen})=\phi(n)$.
The same statement holds for the resident-side master list model, i.e., $\app(\I_{\rm R\mathchar`-ML})=\phi(n)$.
\end{theorem}
We present a full proof of Theorem~\ref{thm:general-approximable} in the Appendix.
It is in large part the same as that of Theorem~7 of Goko et al.~\cite{Goko}, which shows that the approximation factor of {\sc Double Proposal} is $\phi(n)$ in the complete list setting. As shown in the second row of Table~\ref{table1}, {\sc Double Proposal} fails to attain this factor in the current incomplete list setting. To show the approximation factor $\phi(n)$ of {\sc Triple Proposal}, we need to use Lemma~\ref{lem:property2}, i.e., the nonexistence of length-3 augmenting path. Here, we present a proof outline of $\app(\I_{\rm Gen})\leq \phi(n)$. 

Let $M$ be the output of the algorithm and $N$ be an optimal stable matching.
We define vectors $p_M$ and $p_N$ on $R$, which distribute the scores to residents.
For each $h\in H$, among residents in $M(h)$, 
we set $p_M(r)=\frac{1}{\ell(h)}$ for 
$\min\{\ell(h),|M(h)|\}$ residents and $p_M(r)=0$ for the remaining ones.
We also set $p_M(r)=0$ for residents $r$ unmatched in $M$.
Similarly, we define $p_N$ from $N$.
Then, we have $s(M)=\sum_{r\in R} p_M(r)$ and $s(N)=\sum_{r\in R}p_N(r)$.  

In the proof of Theorem 7 in \cite{Goko}, the authors upper bound $\frac{s(N)}{s(M)}$ by comparing the vectors $p_M$ and $p_N$.
Their argument still works for the current incomplete list setting if $s(M)\geq 2$ or there is no resident matched only in $N$.

We need a separate analysis for the remaining case, i.e., the case when \mbox{$s(M)<2$} holds and there exists a resident $r^*$ with $N(r^*)\neq\varnothing= M(r^*)$. Denote $h^*\coloneqq N(r^*)$. 
Since $(r^*, h^*)$ belongs to $E$ and $r^*$ is unmatched in $M$, the stability of $M$ implies that the hospital $h^*$ is full in $M$. By the same argument, for any resident $r$ with $N(r)\neq \varnothing= M(r)$, the hospital $N(r)$ is full in $M$,
which implies $N(r)=h^*$ because there is at most one sufficient hospital in $M$ by $s(M)<2$.
Therefore, any resident $r$ with $N(r)\neq\varnothing= M(r)$ satisfies $r\in N(h^*)$.
Since $h^*$ is the unique sufficient hospital, any resident $r$ with $M(r)\neq \varnothing$ and $p_M(r)=0$ satisfies $r\in M(h^*)$.
Thus, every resident $r$ with $p_M(r)=0$ belongs to $N(h^*)\cup M(h^*)$ unless $r\in S\coloneqq\set{r'\in R|M(r')=N(r')=\varnothing}$. In other words, all residents in $R\setminus (N(h^*)\cup M(h^*)\cup S)$ have positive $p_M$-values.

Since all hospitals other than $h^*$ are deficient in $M$, 
any \mbox{$r\in M(h^*)\setminus N(h^*)$} satisfies $N(r)=\varnothing$, 
because otherwise $(r^*,h^*,r, N(r))$ forms a length-3 $M$-augmenting path, which contradicts Lemma~\ref{lem:property2}.
Therefore, all residents $r$ in \mbox{$M(h^*)\setminus N(h^*)$}, as well as all residents in $S$, satisfy $p_N(r)=0$.
Hence, the summation of $p_N$-values over $N(h^*)\cup M(h^*)\cup S$ is at most $1$, while that of $p_M$-values over $N(h^*)\cup M(h^*)\cup S$ is at least $1$ because $h^*$ is full in $M$.

Combining these observations with others, we upper bound \mbox{$\frac{\sum_{r\in R}p_N(r)}{\sum_{r\in R}p_M(r)}=\frac{s(N)}{s(M)}$}.
\vspace{-1mm}
\begin{toappendix}
\begin{proof}[Proof of Theorem~\ref{thm:general-approximable}]
Here we only show that $\frac{\opt(I)}{\wst(I)}\leq \phi(n)$ holds for any instance $I\in \I_{\rm Gen}$. The other inequality is shown in Proposition~\ref{prop:ML-incomplete-tight} with instances in $\I_{\rm R\mathchar`-ML}$.

Let $M$ be the output of the algorithm and let $N$ be an optimal stable matching.
We define vectors $p_M$ and $p_N$ on $R$, which are distributions of scores of $M$ and $N$ to residents, respectively.
For each hospital $h\in H$, its scores in $M$ and $N$ are 
$s_M(h)=\min\{1,\frac{|M(h)|}{\ell(h)}\}$ and $s_N(h)=\min\{1,\frac{|M(h)|}{\ell(h)}\}$, respectively.
We set $\{p_M(r)\}_{r\in M(h)}$ and $\{p_N(r)\}_{r\in N(h)}$ as follows.
Among residents in $M(h)\cap N(h)$, take $\min\{\ell(h), |M(h)\cap N(h)|\}$ members arbitrarily and set $p_M(r)=p_N(r)=\frac{1}{\ell(h)}$ for them. 
If $|M(h)\cap N(h)|>\ell(h)$, set $p_M(r)=p_N(r)=0$ for the remaining residents in $M(h)\cap N(h)$.
If $|M(h)\cap N(h)|<\ell(h)$, then among $M(h)\setminus N(h)$, 
take $\min\{\ell(h)-|M(h)\cap N(h)|, |M(h)\setminus N(h)|\}$ residents arbitrarily 
and set $p_M(r)=\frac{1}{\ell(h)}$ for them. 
If there still remains a resident $r$ in $M(h)\setminus N(h)$ with undefined $p_M(r)$, 
set $p_M(r)=0$ for all such residents. 
Similarly, define $p_N(r)$ for residents in $N(h)\setminus M(h)$.
We set $p_M(r)=0$ for residents $r$ unmatched in $M$
and set $p_N(r)=0$ for residents $r$ unmatched in $N$.

By definition, for each $h\in H$, we have 
$p_M(M(h))=s_M(h)$ and $p_N(N(h))=s_N(h)$,
where the notation $p_M(A)$ is defined as $p_M(A)=\sum_{r\in A}p_M(r)$ for any $A\subseteq R$ and $p_N(A)$ is defined similarly. 
Since each of $\{M(h)\}_{h\in H}$ and $\{N(h)\}_{h\in H}$ is a subpartition of $R$, we have
\[s(M)=p_M(R),\quad s(N)=p_N(R).\]
Thus, what we have to prove is $\frac{p_N(R)}{p_M(R)}\leq \phi(n)$,
where $n=|R|$.

Note that, for any resident $r\in R$, the condition $M(r)=N(r)$ 
means that $r\in M(h)\cap N(h)$ for some $h\in H$. 
Then, the above construction of $p_M$ and $p_N$ implies the following condition for any $r\in R$,
which will be used in the last part of the proof (in the proof of Claim~\ref{claim:less-than-2}).
\begin{equation}
M(r)=N(r) \implies p_M(r)=p_N(r).\label{eq:implication}
\end{equation}

For the convenience of the analysis below, 
let $R'=\{r'_1,r'_2,\dots,r'_n\}$ be the copy of $R$ and identify $p_N$ as a vector on $R'$. 
Consider a bipartite graph $G=(R,R';F)$, where the edge set $F$ is defined by 
$F=\set{(r_i,r'_j)\in R\times R'|p_M(r_i)\geq p_N(r'_j)}$.
For a matching $X\subseteq F$ 
we denote by $\partial(X)\subseteq R\cup R'$ the set of vertices covered by $X$.
Then, we have $|R\cap\partial(X)|=|R'\cap\partial(X)|=|X|$.
\begin{lemma}\label{lem:more-than-half}
$G=(R,R';F)$ admits a matching $X$ such that $|X|\geq \lceil\frac{n}{2}\rceil$.
Furthermore, in case $s(M)<2$, such a matching $X$ can be chosen so that $p_M(R\cap \partial(X))\geq 1$ holds 
and any $r\in R\setminus \partial(X)$ satisfies $p_M(r)\neq 0$.
\end{lemma}
We postpone the proof of this lemma and now complete the proof of Theorem~\ref{thm:general-approximable}.
There are two cases (i) $s(M)\geq 2$ and (ii) $s(M)<2$.

We first consider the case (i). Assume $s(M)\geq 2$. 
By Lemma~\ref{lem:more-than-half}, there is a matching $X\subseteq F$ such that $|X|\geq \lceil\frac{n}{2}\rceil$.
The definition of $F$ implies $p_M(R\cap\partial(X))\geq p_N(R'\cap\partial(X))$.
We then have $p_N(R')=p_N(R'\cap \partial(X))+p_N(R'\setminus \partial(X))\leq p_M(R\cap \partial(X))+p_N(R'\setminus \partial(X))
=\{p_M(R)-p_M(R\setminus\partial(X))\}+p_N(R'\setminus \partial(X))$,
which implies the first inequality of the following consecutive inequalities, where others are explained below. 
\begin{eqnarray*}
\frac{s(N)}{s(M)}
&=&\frac{p_N(R')}{p_M(R)}
\\
&\leq&
\frac{p_M(R)-p_M(R\setminus\partial(X))+p_N(R'\setminus \partial(X))}{p_M(R)}\\
&\leq&\frac{p_M(R)+|R'\setminus \partial(X)|}{p_M(R)}\\
&\leq&\frac{2+|R'\setminus \partial(X)|}{2}\\
&\leq&\frac{2+\lfloor\frac{n}{2}\rfloor}{2}\\
&\leq&\phi(n).
\end{eqnarray*}
The second inequality uses the facts that $p_M(r)\geq 0$ for any $r\in R$ and $p_N(r')\leq 1$ for any $r'\in R'$.
The third follows from $p_M(R)=s(M)\geq 2$. The fourth follows from $|X|\geq \lceil\frac{n}{2}\rceil$ as it implies
$|R'\setminus \partial(X)|=|R'|-|X|\leq n-\lceil\frac{n}{2}\rceil=\lfloor\frac{n}{2}\rfloor$.
The last one $\frac{2+\lfloor\frac{n}{2}\rfloor}{2}\leq \phi(n)$ can be checked for $n=1,2$ easily and
for $n\geq 3$ as follows:
\[\phi(n)-\frac{2+\lfloor\frac{n}{2}\rfloor}{2}=\frac{n(1+\lfloor\frac{n}{2}\rfloor)}{n+\lfloor \frac{n}{2}\rfloor}-\frac{2+\lfloor\frac{n}{2}\rfloor}{2}
=\frac{\lfloor\frac{n}{2}\rfloor(n-2-\lfloor\frac{n}{2}\rfloor)}{2(n+\lfloor \frac{n}{2}\rfloor)}
=\frac{\lfloor\frac{n}{2}\rfloor(\lceil\frac{n}{2}\rceil-2)}{2(n+\lfloor \frac{n}{2}\rfloor)}\geq 0.\]
Thus, we obtain $\frac{s(N)}{s(M)}\leq \phi(n)$ as required.

We next consider the case (ii). Assume $s(M)<2$. By Lemma~\ref{lem:more-than-half},
then there is a matching $X\subseteq F$ such that $|X|\geq \lceil\frac{n}{2}\rceil$, $p_M(R\cap \partial(X))\geq 1$, 
and $p_M(r)\neq 0$ for any $r\in R\setminus \partial(X)$.
Again, by the definition of $F$, we have $p_M(R\cap\partial(X))\geq p_N(R'\cap \partial(X))$,
which implies the first inequality of the following consecutive inequalities, where the others are explained below. 
\begin{eqnarray*}
\frac{s(N)}{s(M)}
&=&\frac{p_N(R')}{p_M(R)}\\
&\leq&\frac{p_M(R\cap \partial(X))+p_N(R'\setminus \partial(X))}{p_M(R\cap \partial(X))+p_M(R\setminus \partial(X))}\\
&\leq&\frac{p_M(R\cap \partial(X))+|R'\setminus \partial(X)|}{p_M(R\cap \partial(X))+\frac{1}{n}|R\setminus \partial(X)|}\\
&\leq&\frac{1+|R'\setminus \partial(X)|}{1+\frac{1}{n}|R\setminus \partial(X)|}\\
&\leq&\frac{1+\lfloor\frac{n}{2}\rfloor}{1+\frac{1}{n}\lfloor\frac{n}{2}\rfloor}\ = \ \phi(n).
\end{eqnarray*}
The second inequality follows from the facts that $p_N(r')\leq 1$ for any $r'\in R'$ and $p_M(r)\neq 0$ for any $r\in R\setminus \partial(X)$.
Note that $p_M(r)\neq 0$ implies $p_M(r)=\frac{1}{\ell(h)}\geq \frac{1}{n}$ where $h\coloneqq M(r)$.
The third follows from $p_M(R\cap \partial(X))\geq 1$.
The last one follows from $|R'\setminus \partial(X)|=|R\setminus \partial(X)|=n-|X|\leq \lfloor\frac{n}{2}\rfloor$.
Thus, we obtain $\frac{s(N)}{s(M)}\leq \phi(n)$ also for this case.
\switch{\qed}{}\end{proof}
\begin{proof}[Proof of Lemma~\ref{lem:more-than-half}]
To show the first claim of the lemma, 
we intend to construct a matching in $G$ of size at least $\lceil\frac{n}{2}\rceil$.
We need some preparation for this construction.

Partition the set $R$ of residents into three subsets:
\begin{align*}
&R_{+}\coloneqq \set{r\in R|M(r)\succ_{r} N(r)},\\
&R_{-}\coloneqq \set{r\in R|N(r)\succ_{r} M(r) \text{ or }  [M(r)=_{r}N(r), ~p_N(r)>p_M(r)]}, \mbox{ and}\\
&R_{0}~\!\coloneqq \set{r\in R|M(r)=_{r}N(r), ~p_M(r)\geq p_N(r)}.
\end{align*}
Let $R'_+$, $R'_-$, and $R'_0$ be the corresponding subsets of $R'$. 

\begin{\switch{myclaim}{claim}}\label{claim:injection}
There is an injection $\xi_+\colon R_+\to R'$ such that $p_M(r)=p_N(\xi_+(r))$ for every $r\in R_+$.
There is an injection $\xi_-\colon R'_-\to R$ such that $p_N(r')=p_M(\xi_-(r'))$ for every $r'\in R'_-$.
\end{\switch{myclaim}{claim}}
\begin{proof}
We first construct $\xi_+\colon R_+\to R'$. Set $M(R_+)\coloneqq \set{M(r)|r\in R_+}$. 
For each hospital $h\in M(R_+)$, any $r\in M(h)\cap R_+$ satisfies $h=M(r)\succ_{r}N(r)$.
By the stability of $N$, then $h$ is full in $N$. 
Therefore, in $N(h)$, there are $\ell(h)$ residents with $p_N$ value $\frac{1}{\ell(h)}$ 
and $u(h)-\ell(h)$ residents with $p_N$ value $0$.
Since $|M(h)|\leq u(h)$ and $p_M$ values are $\frac{1}{\ell(h)}$ for $\min\{|M(h)|, \ell(h)\}$ residents, we can define
an injection $\xi_+^h\colon M(h)\cap R_+\to N(h)$ such that $p_M(r)=p_N(\xi_+^h(r))$
for every $r\in M(h)\cap R_+$.
By regarding $N(h)$ as a subset of $R'$ and taking the direct sum of $\xi_+^h$ for all $h\in M(R_+)$,
we obtain an injection $\xi_+\colon R_+\to R'$ such that $p_M(r)=p_N(\xi_+(r))$ for every $r\in R_+$.

We next construct $\xi_-\colon R'_-\to R$.
Define $N(R'_-)\coloneqq \set{N(r')|r'\in R'_-}$.
For each $h'\in N(R'_-)$, any resident $r\in N(h')\cap R'_-$ satisfies either 
$h'=N(r)\succ_{r} M(r)$ or $[h'=N(r)=_{r}M(r), ~p_N(r)>p_M(r)]$.
In case that some resident $r$ satisfies $h'=N(r)\succ_{r} M(r)$,
the stability of $M$ implies that $h'$ is full in $M$. 
Then, we can define an injection $\xi_-^{h'}\colon N(h')\cap R'_-\to M(h')$ in the manner we defined $\xi_+^{h}$ above and
$p_N(r')=p_M(\xi_-^{h'}(r'))$ holds for any $r'\in N(h')\cap R'_-$.
We then assume that all residents $r\in N(h')\cap R'_-$ satisfy $[h'=N(r)=_{r}M(r), ~p_N(r)>p_M(r)]$.
Then, all those residents satisfy $0\neq p_N(r)=\frac{1}{\ell(h')}$,
and hence $|N(h')\cap R'_-|\leq \ell(h')$. Note that $N(r)=_{r}M(r)$ implies $M(r)\neq \varnothing$ and let $h\coloneqq M(r)$. 
Then by $p_N(r)>p_M(r)$, we have either $\ell(h')<\ell(h)$ or $p_M(r)=0$, where the latter implies $|M(h)|>\ell(h)$.
As we have $h=_{r}h'$ and $M(r)=h$, by Lemma~\ref{lem:property}, 
each of $\ell(h')<\ell(h)$ and $|M(h)|>\ell(h)$ implies $|M(h')|\geq \ell(h')$,
and hence there are $\ell(h')$ residents whose $p_M$ values are $\frac{1}{\ell(h')}$.  
Since $p_N(r)=\frac{1}{\ell(h')}$ for all residents $r\in N(h')\cap R'_-$, 
we can define an injection $\xi_-^{h'}\colon N(h')\cap R'_-\to M(h')$ such that $p_N(r')=p_M(\xi_-^{h'}(r'))=\frac{1}{\ell(h')}$ for every $r'\in N(h')\cap R'_-$.
By taking the direct sum of $\xi_-^{h'}$ for all $h'\in M(R'_-)$,
we obtain an injection $\xi_-\colon R'_-\to R$ such that $p_N(r')=p_M(\xi_-(r'))$ for every $r'\in R'_-$.
\switch{\qed}{}\end{proof}

We now define a bipartite (multi-)graph.
Let $G^*=(R,R';F^*)$, where $F^*$ is the disjoint union of $F_+$, $F_-$, and $F_0$, which are defined as
\begin{align*}
F_+&\coloneqq \set{(r,\xi_+(r))|r\in R_+},\\
F_-&\coloneqq \set{(\xi_-(r'),r')|r\in R'_-},  \mbox{ and}\\
F_0~\!&\coloneqq \set{(r,r')|r\in R_0 \text{ and $r'$ is the copy of $r$}}.
\end{align*}
See Figure~\ref{fig:graph1-second} for an example (which is taken from \cite{goko_et_al:LIPIcs.STACS.2022.31}).
Note that $F^*$ can have multiple edges between $r$ and $r'$ if $(r,r')=(r,\xi_+(r))=(\xi_-(r'),r')$. 
By the definitions of $\xi_+$, $\xi_-$, and $R_0$, any edge $(r,r')$ in $F^*$ satisfies $p_M(r)\geq p_N(r')$.
Since $F=\set{(r,r')|p_M(r)\geq p_N(r')}$, any matching in $G^*$ is also a matching in $G$.

\begin{figure}[t]
	\begin{center}
		\includegraphics[width=0.43\hsize]{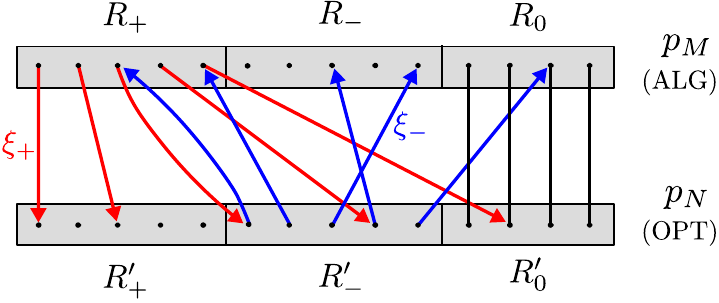}
	\end{center}
	\caption{\small A graph $G^*=(R,R';F^*)$. The upper and lower rectangles represent $R$ and $R'$, respectively. 
	The edge sets $F_+$, $F_-$, and $F_0$
	are respectively represented by downward directed edges, upward directed edges, 
	and undirected edges.}
	\label{fig:graph1-second}
\end{figure}

Then, the following claim completes the first statement of the lemma.
\begin{\switch{myclaim}{claim}}
$G^*$ admits a matching whose size is at least $\lceil\frac{n}{2}\rceil$, and so does $G$.
\end{\switch{myclaim}{claim}}
\begin{proof}
Since $\xi_+\colon R_+\to R'$ and $\xi_-\colon R'_-\to R$ are injections, 
every vertex in $G^*$ is incident to at most two edges in $F^*$ as follows:
Each vertex in $R_+$ (resp., $R'_-$) is incident to exactly one edge in $F_+$ (resp., $F_-$) and at most one edge in $F_-$ (resp., $F_+$).
Each vertex in $R_-$ (resp., $R'_+$) is incident to at most one edge in $F_-$ (resp., $F_+$).
Each vertex in $R_0$ (resp., $R'_0$) is incident to exactly one edge in $F_=$ and at most one edge in $F_-$ (resp., $F_+$).

Since $F^*$ is the disjoint union of $F_+$, $F_-$, and $F_0$, we have $F^*=|F_+|+|F_-|+|F_=|=|R_+|+|R_-|+|R_0|=n$.
As every vertex is incident to at most two edges in $F^*$,
each connected component $K$ of $G^*$ forms a path or a cycle.
In $K$, we can take a matching that contains at least a half of the edges in $K$. 	
(Take edges alternately along a path or a cycle. 
For a path with odd edges, let the end edges be contained.)
The union of such matchings in all components forms a matching
in $G^*$ whose size is at least $\lceil\frac{n}{2}\rceil$.
\switch{\qed}{}\end{proof}
In the rest, we show the second claim of Lemma~\ref{lem:more-than-half}. 
For this purpose, we make the following observation.
Suppose that there is a matching $Y$ in $G$. 
Then, there is a maximum matching $X$ in $G$ such that $\partial(Y)\subseteq \partial(X)$.
This follows from the behavior of the augmenting path algorithm to compute a maximum matching in a bipartite graph (see e.g., \cite{schrijver2003combinatorial}).
In this algorithm, a matching, say $X$, is repeatedly updated to reach the maximum size. 
Throughout the algorithm, $\partial(X)$ is monotone increasing. Therefore, if 
we initialize $X$ by $Y$, it finds a maximum matching with $\partial(Y)\subseteq \partial(X)$.
Additionally, note that $\partial(Y)\subseteq \partial(X)$ implies $p_M(R\cap \partial(Y))\leq p_M(R\cap \partial(X))$
as $p_M$ is an nonnegative vector.
Therefore, the following claim completes the proof of the second claim of Lemma~\ref{lem:more-than-half}.

\begin{\switch{myclaim}{claim}}\label{claim:less-than-2}
If $s(M)<2$, then there is a matching $Y$ in $G$ such that $p_M(R\cap \partial(Y))\geq 1$ holds 
and any $r\in R\setminus \partial(Y)$ satisfies $p_M(r)\neq 0$. 
\end{\switch{myclaim}{claim}}
\begin{proof}
Let $S\coloneqq\set{r\in R|M(r)=N(r)=\varnothing}$ and let $S'$ be the subset of $R'$ corresponding to $S$. Let $Y_1\coloneqq \set{(r,r')\in S\times S'|\text{$r'$ is the copy of $r$}}$.
Because $p_M(r)=p_N(r)=0$ for any $r\in S$, we have $Y_1\subseteq F$ and it is a matching in $G$.

We now show that there is a matching $Y_2\subseteq F$ that is vertex-disjoint from $Y_1$ and satisfies $p_M(R\cap \partial(Y_2))\geq 1$
and $\set{r\in R\setminus S|p_M(r)=0} \subseteq R\cap \partial(Y_2)$.
For such a matching $Y_2$, we see that $Y\coloneqq Y_1\cup Y_2$ forms a matching satisfying the conditions required in the claim.
Therefore, constructing such a matching $Y_2$ completes the proof.
We consider the following three cases.
\begin{description}
\item[Case 1.] Some resident in $R\setminus S$ is unmatched in $M$.
\item[Case 2.] All residents in $R\setminus S$ are matched in $M$ but some $r\in R\setminus S$ satisfies $p_M(r)=0$.
\item[Case 3.] All residents in $R\setminus S$ satisfy $p_M(r)>0$.
\end{description}

{\bf Case 1.} 
Take a resident $r^*\in R\setminus S$ satisfying $M(r^*)=\varnothing$.
As $r^*\not\in S$, resident $r$ is matched in $N$. Denote $h^*=N(r^*)$.
Since $(r^*, h^*)$ belongs to $E$ and $r^*$ is unmatched in $M$, the stability of $M$ implies that the hospital $h^*$ is full in $M$. By the same argument, for any resident $r\in R\setminus S$ with $M(r)=\varnothing$, the hospital $N(r)$ is full,
which implies $N(r)=h^*$ because there is at most one sufficient hospital in $M$ by $s(M)<2$.
Therefore, any resident $r\in R\setminus S$ with $M(r)=\varnothing$ satisfies $r\in N(h^*)$.
Since $h^*$ is the unique sufficient hospital, any resident $r$ with $M(r)\neq \varnothing$ and $p_M(r)=0$ satisfies $r\in M(h^*)$.
Thus, every resident $r\in R\setminus S$ with $p_M(r)=0$ belongs to $N(h^*)\cup M(h^*)$.

Since all hospitals other than $h^*$ are deficient in $M$, 
any $r\in M(h^*)\setminus N(h^*)$ satisfies $N(r)=\varnothing$, 
because otherwise $(r^*,h^*,r, N(r))$ forms a length-3 $M$-augmenting path, which contradicts Lemma~\ref{lem:property2}.
Then,  every $r\in M(h^*)\setminus N(h^*)$ satisfies $p_N(r)=0$.
Hence, in $N(h^*)\cup M(h^*)$,  at most $\ell(h^*)$ residents satisfy $p_N(r)=\frac{1}{\ell(h^*)}$ and all other residents satisfy $p_N(r)=0$.
At the same time, because $h^*$ is full in $M$, at least $\ell(h^*)$ residents satisfy $p_M(r)=\frac{1}{\ell(h^*)}$.
Therefore, there exists a bijection $\xi$ from $N(h^*)\cup M(h^*)$ to its copy in $R'$ such that $p_M(r)\geq p_N(\xi(r))$ for every $r\in N(h^*)\cup M(h^*)$.
By setting $Y_2\coloneqq \set{(r,\xi(r))|r\in N(h^*)\cup M(h^*)}$, we can obtain a required matching.
\medskip

{\bf Case 2.} 
Let $r^*\in R\setminus S$ be a resident with $p_M(r^*)=0$ while $M(r^*)\neq \varnothing$.
Then, $h\coloneqq M(r^*)$ satisfies $|M(h)|>\ell(h)$, and hence $p_M(M(h))=1$.
Since $s(M)<2$, any hospital other than $h$ should be deficient.
Therefore, the value of $p_M$ can be $0$ only for residents in $M(h)$.
\begin{itemize}
\setlength{\itemsep}{2mm}
\medskip
\item If $M(h)\cap R_+\neq \emptyset$, then as shown in the proof of Claim~\ref{claim:injection}, 
$h$ is full in $N$. Then, there is an injection $\xi:M(h)\to N(h)$ such that $p_M(r)=p_N(\xi(r))$ for any $r\in M(h)$.
Hence, $Y_2\coloneqq \set{(r,\xi(r))|r\in M(h)}\subseteq F$ is a matching in $G$ satisfying 
the required conditions.
\item If $M(h)\cap R_-\neq\emptyset$, take any $r\in M(h)\cap R_-$ and set $h'\coloneqq N(r)$.
As shown in the proof of Claim~\ref{claim:injection}, then $|M(h')|\geq \ell(h')$,
i.e., $h'$ is sufficient. 
Note that we have $p_M(r)\neq p_N(r)$ only if $M(r)\neq N(r)$ by the condition \eqref{eq:implication}.
Since $r\in R_-$ implies either $h'\succ_r h$ or $p_N(r)>p_M(r)$, we have $h'\neq h$,
which contradicts the fact that $h$ is the unique sufficient hospital.
Thus, $M(h)\cap R_-\neq\emptyset$ cannot happen.
\item If $M(h)\cap R_+=\emptyset$ and $M(h)\cap R_-=\emptyset$, we have $M(h)\subseteq R_0$.
Then, by setting $Y_2\coloneqq \set{(r,r')|r\in M(h), ~\text{$r'$ is the copy of $r$}}$, we can obtain a required matching.
\end{itemize}

{\bf Case 3.} 
We next consider the case where $p_M(r)\neq 0$ for any $r\in R\setminus S$.
Then, our task is to find a matching $Y_2\subseteq F$ that is vertex-disjoint from $Y_1$ and satisfies $p_M(R\cap \partial(Y_2))\geq 1$.
In this case, for any resident $r\in R$ with $h=M(r)$, we always have $p_M(r)=\frac{1}{\ell(h)}$. 
\begin{itemize}
\setlength{\itemsep}{2mm}
\medskip
\item 
If $R_+\neq \emptyset$, then $M(R_+)\coloneqq \set{M(r)|r\in R_+}\neq \emptyset$.
Since $\sum_{h\in M(R_+)}|M(h)\cap R_+|=|R_+|$ and $\sum_{h\in M(R_+)}|N(h)\cap R_+|\leq |R_+|$, 
there is at least one hospital $h\in M(R_+)$ such that $|M(h)\cap R_+|\geq |N(h)\cap R_+|$.
Let $h$ be such a hospital. Since $h\in M(R_+)$, as shown in the proof of Claim~\ref{claim:injection},
$h$ is full in $N$ and there are $\ell(h)$ residents $r$ with $p_N(r)=\frac{1}{\ell(h)}$.
We intend to show that there are at least $\ell(h)$ residents $r$ with $p_M(r)\geq \frac{1}{\ell(h)}$,
which implies the existence of a required $Y$.
Regard $N(h)$ as a subset of $R'$. 
If there is some $r'\in N(h)\cap R'_-$, as seen in the proof of Claim~\ref{claim:injection},
there are $\ell(h)$ residents $r$ with $p_M(r)=\frac{1}{\ell(h)}$, and we are done.
So, assume $N(h)\cap R'_-=\emptyset$, which implies $N(h)\subseteq R'_+\cup R'_0$.
Since $|M(h)\cap R_+|\geq |N(h)\cap R_+|$, at least $|N(h)\cap R_+|$ residents in $R_+$ belongs to $M(h)$.
As $p_M$ is positive, then at least $|N(h)\cap R_+|$ residents $r\in M(h)$ satisfy $p_M(r)=\frac{1}{\ell(h)}$.
Additionally, by the definition of $F_0$, each $r\in N(h)\cap R_0$ satisfies $p_M(r)\geq p_N(r)$,
where $p_N(r)=\frac{1}{\ell(h)}$ for at least $\ell(h)-|N(h)\cap R_+|$ residents in $N(h)\cap R_0$.
Thus, at least $\ell(h)$ residents $r\in R$ satisfy $p_M(r)\geq \frac{1}{\ell(h)}$.

\item If $R_-\neq \emptyset$, then $N(R_-)\coloneqq \set{N(r)|r\in R_-}\neq \emptyset$.
Similarly to the argument above, 
there is at least one hospital $h\in N(R_-)$ such that $|N(h)\cap R_-|\geq |M(h)\cap R_-|$.
Let $h$ be such a hospital. Since $h\in N(R_-)$, as shown in the proof of Claim~\ref{claim:injection},
$h$ is sufficient in $M$ and there are $\ell(h)$ residents $r$ with $p_M(r)=\frac{1}{\ell(h)}$.
We intend to show that there are at least $\ell(h)$ residents $r$ with $p_N(r)\leq \frac{1}{\ell(h)}$.
If there is some $r\in M(h)\cap R_+$, we are done as in the previous case.
So, assume $M(h)\cap R_+=\emptyset$, which implies $M(h)\subseteq R_-\cup R_0$.
Since $|N(h)\cap R_-|\geq |M(h)\cap R_-|$, at least $|M(h)\cap R_-|$ residents $r$ in $R_-$ belongs to $N(h)$,
and satisfies $p_N(r)\in \{\frac{1}{\ell(h)},0\}$.
Additionally, by the definition of $F_0$, each $r\in M(h)\cap R_0$ satisfies $p_N(r)\leq p_M(r)=\frac{1}{\ell(h)}$.
Thus, at least $\ell(h)$ residents $r\in R$ satisfy $p_N(r)\leq \frac{1}{\ell(h)}$.

\item If $R_+=\emptyset$ and $R_-=\emptyset$, then $R_0=R$ and $F_0$ forms a matching and we have $\partial(F)\cap R=R$.
Since $p_M(r)=\frac{1}{\ell(h)}\geq \frac{1}{n}$ for any $h\in H$ and $r\in M(h)$, we have $p_M(\partial(F)\cap R)\geq 1$.
\end{itemize}
Thus, in any case, we can find a matching with required conditions.
\switch{\qed}{}\end{proof}
Thus we completed the proof of the second claim of the lemma.
\switch{\qed}{}\end{proof}

\begin{proposition}\label{prop:ML-incomplete-tight}
For any natural number $n$, 
there is an instance $I \in \I_{\rm R\mathchar`-ML}$ with $n$ residents such that $\frac{\opt(I)}{\alg(I)}\geq \phi(n)$.
\end{proposition}
\begin{proof}
The theorem follows from a simple modification of the proof of Proposition~27 in Goko et al.~\cite{Goko}, which shows a lower bound on the approximation factor of {\sc Double Proposal}.
The idea is that if we remove irrelevant hospitals from residents' preference lists, then we can define a master list so that all the resulting residents' preference lists are consistent with it.

The case of $n=1$ is trivial as argued in \cite{Goko}.

For $n=2$, we consider an instance $I$ defined in the last part of the proof of Theorem~\ref{thm:approx_mar}.

For $n\geq 3$, we define $I$ as follows: 
The set of residents is $R=R'\cup R''$, where $R'=\{r'_1,r'_2,\dots,r'_{\lceil\frac{n}{2}\rceil}\}$
and $R''=\{r''_1,r''_2,\dots,r''_{\lfloor\frac{n}{2}\rfloor}\}$, and 
the set of hospital is $H=\{h_1,h_2\dots, h_n\}\cup\{x,y\}$.
The preference lists are given as follows.
Here ``(~~$R$~~)'' represents the tie consisting of all residents
and  ``[~~$R''$~~]'' denotes an arbitrary strict order on $R''$.
\begin{center}
\renewcommand\arraystretch{1.2}
\begin{tabular}{llllllllllllllllll}
$r'_{i}$: & $x$ & $h_i$ & &\hspace{15mm} & $x$ $[\lceil\frac{n}{2}\rceil, \lceil\frac{n}{2}\rceil]$: &  (~~~~~$R$~~~~~) \hspace{15mm}\\
$r''_{i}$: & $x$ & $y$ & &\hspace{15mm} & $y$ $[n, n]$: &  [~~$R''$~~]\\
 & & &  & & $h_i$ $[1, 1]$: & $r'_{i}$ \\
\end{tabular}
\end{center}
Execute {\sc Triple Proposal} for $I$ by prioritizing residents in $R'$ over those in $R''$ whenever an arbitrariness arises. 
That is, if a resident in $R'$ and that in $R''$ both satisfy the condition to be selected as $r'$ at Line~\ref{chosen2}, we choose a resident in $R''$. In case $n=2$, let $r_1$ be prioritized over $r_2$.
By the same arguments as in the proof of Proposition~27 in \cite{Goko}, we obtain $\frac{\opt(I)}{\alg(I)}\geq \phi(n)$.
\switch{\qed}{}\end{proof}
\end{toappendix}

\section{Inapproximability Results}
\switch{We provide inapproximability results for HRT-MSLQ. The first subsection investigates the uniform model and the second one investigates the other models.}
{This section is divided into three subsections.
In the first one, we show inapproximability results for the uniform model using hardness results for the minimum vertex cover problem. 
For other models, i.e., the general, marriage, and $R$-side master list models, the second subsection provides inapproximability results, which follow from previous works.
The third subsection is devoted to the postponed proof of a theorem on the vertex cover problem used in the first subsection.}
\subsection{Uniform Model}
This section shows the following inapproximability results for the uniform model.

\begin{theorem}\label{thm:inapprox-uniform}
HRT-MSLQ in the uniform model with $\theta=\frac{u}{\ell}$ is 
\begin{itemize}
\switch{\setlength{\leftskip}{3mm}}{}
\item[\rm (i)] inapproximable within a factor $\frac{1}{3}\theta+1- \epsilon\theta^2$ for any $\epsilon>0$ under UGC, and
\item[\rm (ii)] inapproximable within a factor $\frac{2-\sqrt{2}}{3}\theta+1-\epsilon\theta^2$ for any $\epsilon>0$ unless P=NP.
\end{itemize}
\end{theorem}

We prove Theorem~\ref{thm:inapprox-uniform} by extending the idea of Yanagisawa~\cite{Yanagisawa}, who provided inapproximability results for MAX-SMTI by a reduction from the Vertex Cover problem (VC for short). 
VC is a well-studied NP-hard problem that asks to find a minimum size vertex cover in a given graph. 
We denote by VC-PM the restricted version of VC where an input graph is imposed to have a perfect matching. 

Yanagisawa \cite{Yanagisawa} used inapproximability results for VC-PM under UGC \cite{KR08,CC07} and under P$\neq$NP~\cite{DS05} to obtain inapproximability results for MAX-SMTI.
To prove Theorem~\ref{thm:inapprox-uniform}(i), we use the following result as in \cite{Yanagisawa}.

\begin{theorem}[Chleb{\'\i}k and Chleb{\'\i}kov{\'a} \cite{CC07}, Khot and Regev \cite{KR08}]
\label{theorem-vcpm-ugc}
Under UGC, VC-PM is inapproximabile within a factor $2-\epsilon$ for any $\epsilon>0$.
\end{theorem}
To prove Theorem~\ref{thm:inapprox-uniform}(ii), we show the following hardness result for \mbox{VC-PM}, which can be obtained by a recent result for VC \cite{Khot18}, together with the reduction from VC to VC-PM in \cite{CC07}. 
For a graph $G$, we denote by $\optvc(G)$ the size of a minimum vertex cover of $G$.

\begin{theorem}\label{thm:VC-PM-gap}
For any $\epsilon>0$, it is NP-hard for graphs $G=(V,E)$ with perfect matchings to distinguish between the following two cases: 
\begin{itemize}
\switch{\setlength{\leftskip}{3mm}}{}
\item[\rm (1)] $\optvc(G) \leq (\frac{\sqrt{2}}{2}+\epsilon)|V|$.
\item[\rm (2)]$\optvc(G) \geq (1-\epsilon)|V|$. 
\end{itemize}
\end{theorem}
The proof of Theorem~\ref{thm:VC-PM-gap} is postponed to Appendix~\ref{sec:postponed}.

To transform the above two hardness results for VC-PM into those for HRT-MSLQ, 
we provide a method to transform VC-PM instances to 
HRT-MSLQ instances satisfying the following properties. 
\begin{proposition}
\label{prop-hard--1x}
Let $\ell$ and $u$ be positive integers with $\ell \leq u$ and set $\theta=\frac{u}{\ell}$.
There is a polynomial-time reduction from VC-PM instances $G=(V,F)$ to HRT-MSLQ instances $I$ in the uniform model with quotas $[\ell,u]$ such that
\begin{itemize}
\switch{\setlength{\leftskip}{3mm}}{}
\item[\rm (a)] for any stable matching $M$ in $I$, there exists a vertex cover $C$ of $G$ satisfying \switch{}{the condition}
$s(M)\leq (1.5+\theta)|V|-\theta|C|$, and
\item[\rm (b)]  $\opt(I)=(1.5+\theta)|V|-\theta\optvc(G)$.
\end{itemize}
\end{proposition}
The proof of Proposition~\ref{prop-hard--1x} is presented in Appendix~\ref{sec:postponed2}. 

We now complete the proof of Theorem~\ref{thm:inapprox-uniform} relying on this proposition.
\begin{proof}[\bf {\em Proof of Theorem~\ref{thm:inapprox-uniform}}]
We show the statements (i) and (ii) in Theorem~\ref{thm:inapprox-uniform}, which are inapproximability of HRT-MSLQ under UGC and under P$\neq$NP, respectively.
Let $G$ be a VC-PM instance and $I$ be an HRT-MSLQ instance satisfying the conditions (a) and (b) in Proposition~\ref{prop-hard--1x}.

We first show (i). For any $\epsilon>0$, let $r=\frac{\theta}{3}+1- \epsilon\theta^2$ and 
let $M$ be a stable matching in $I$ with $s(M) \geq \frac{\opt(I)}{r}$. 
By Proposition \ref{prop-hard--1x}(a), 
there exists a vertex cover $C$ of $G$ such that 
$s(M)\leq (1.5+\theta)|V|-\theta|C|$. 
Thus we have 
\begin{eqnarray*}
|C|&\leq &\frac{1.5+\theta}{\theta}|V|-\frac{\opt(I)}{r\theta}\ \leq \ \biggl(1-\frac{1}{r}\biggr)\frac{1.5+\theta}{\theta}|V|+\frac{\optvc(G)}{r}, 
\end{eqnarray*}
where the second inequality follows from $\opt(I)=(1.5+\theta)|V|-\theta\optvc(G)$ by Proposition~\ref{prop-hard--1x}(b). 
This implies that 
\begin{eqnarray*}
\frac{|C|}{\optvc(G)}&\leq &
\frac{1}{r}+
 \biggl(1-\frac{1}{r}\biggr)\frac{1.5+\theta}{\theta}\frac{|V|}{\optvc(G)}\\
 &\leq &   \frac{1}{r}+2
 \biggl(1-\frac{1}{r}\biggr)\frac{1.5+\theta}{\theta} \\ 
 &= &   2+\frac{3}{\theta} -\frac{\theta+3}{r\theta} \\ 
 &\leq & 2-\frac{9\epsilon \theta}{\theta+3} \hspace*{3.5cm}  \textstyle{(\mbox{by }\frac{1}{r}\geq \frac{3(\theta+3+3\epsilon \theta^2)}{(\theta+3)^2})}\\
 &\leq& 2-\epsilon  \hspace*{4.25cm} (\mbox{by }\theta\geq 1), 
\end{eqnarray*}
where the second inequality follows from $\optvc(G)\geq 0.5|V|$, since $G$ has a perfect matching. 
Therefore, by Theorem \ref{theorem-vcpm-ugc}, 
the uniform model of HRT-MSLQ is 
inapproximable within a factor $\frac{\theta}{3}+1- \epsilon\theta^2$ for any positive $\epsilon$ under UGC. 

\medskip

We next show (ii).
 By Proposition \ref{prop-hard--1x}(b), we have 
 $\opt(I)=(1.5+\theta)|V|-\theta\optvc(G)$.  If $\optvc(G)\leq (\frac{\sqrt{2}}{2}+\delta)|V|$ for some positive $\delta$, 
 then  $\opt(I)\geq \frac{1}{2} ((2-\sqrt{2}-2\delta)\theta+ 3)|V|$ holds.  
 On the other hand, if 
 $\optvc(G)\geq (1-\delta)|V|$ for some positive $\delta$, 
 we have   $\opt(I)\leq \frac{1}{2}(3+2\delta\theta)|V|$.  
 
 By Theorem \ref{thm:VC-PM-gap}, 
 it is NP-hard for HRT-MSLQ in the uniform model to approximate within a factor
 \begin{eqnarray*}
 \frac{(2-\sqrt{2}-2\delta)\theta+ 3}{3+2\delta\theta}& > & \frac{2-\sqrt{2}}{3}\theta +1 -2\delta \theta^2, 
 \end{eqnarray*}
where the inequality follows from that $\theta\geq 1$. 
Thus by setting  $\epsilon =2\delta$, 
we obtain the desired inapproximation factor of the problem. 
\switch{\qed}{}\end{proof}

\begin{remark}
By modifying the gadget construction in the proof of \mbox{Proposition~\ref{prop-hard--1x}}, we can improve a lower bound on the approximation factor of the complete list setting of HRT-MSLQ in the uniform model from constant \cite{Goko} to linear in $\theta$.
\end{remark}

\subsection{Other Models}
Here we present some inapproximability results on models other than the uniform model.
The inapproximability results in Table~\ref{table1} for the general and resident-side master list models are consequences of a result in \cite{goko_et_al:LIPIcs.STACS.2022.31} as can be seen from Remark~\ref{rem:ML-inapprox}.

\begin{remark}\label{rem:ML-inapprox}
It is shown in Goko et al.~\cite{goko_et_al:LIPIcs.STACS.2022.31} that, even in the complete list setting, HRT-MSLQ is inapproximable within a factor $n^{\frac{1}{4}-\epsilon}$ for any $\epsilon > 0$ unless P=NP. 
This immediately gives the result in the bottom-left corner of Table~\ref{table1}.
We can modify their proof to obtain the same inapproximability result in the resident-side master list model with incomplete lists.
In the proof of Theorem 10 in \cite{goko_et_al:LIPIcs.STACS.2022.31}, agents denoted by ``$\cdots$'' are added to preference lists in Figure 3 in \cite{goko_et_al:LIPIcs.STACS.2022.31}.
They are actually added only to make the preference lists complete and do not play any role in the reduction.
Actually, removing them still retains the correctness of the reduction, and as a result, we can define a master list from which any resultant resident's preference list can be derived.
The result in the bottom-right corner of Table \ref{table1} is obtained in this manner.
\end{remark}

We then move to the marriage model.
Since the marriage model is a special case of the uniform model where $u=1$ and $\ell$ is either 0 or 1, we can obtain inapproximability results for the marriage model by letting $u = \ell = 1$ (i.e., $\theta=1$) in Theorem~\ref{thm:inapprox-uniform}.

\begin{proposition}\label{prop:marriage-inapprox}
HRT-MSLQ in the marriage model is 
\begin{itemize}
\switch{\setlength{\leftskip}{3mm}}{}
\item[\rm (i)] inapproximable within a factor $\frac{4}{3}- \epsilon$ for any $\epsilon>0$ under UGC, and
\item[\rm (ii)] inapproximable within a factor $\frac{5-\sqrt{2}}{3}-\epsilon$ for any $\epsilon>0$ unless P=NP.
\end{itemize}
\end{proposition}

\begin{remark}\label{rem:SMTI}
Observe that the marriage model when $\ell=u=1$ is equivalent to MAX-SMTI.
Therefore, Proposition \ref{prop:marriage-inapprox}(ii) implies a lower bound $\frac{5-\sqrt{2}}{3}~(\approx 1.1952)$ on the approximation factor of MAX-SMTI, which improves the best known lower bound 
$\frac{33}{29}~(\approx 1.1379)$\cite{Yanagisawa}. 
We note that the same lower bound can also be obtained by directly applying Theorem \ref{thm:VC-PM-gap} to Yanagisawa's reduction \cite{Yanagisawa}.
We also note that, by the same observation, a lower bound on the approximation factor of MAX-SMTI with one-sided-ties can be improved from $\frac{21}{19}~(\approx 1.1052)$ to $\frac{6-\sqrt{2}}{4}~(\approx 1.1464)$.
\end{remark}

\sloppy
\begin{remark}
As mentioned above, MAX-SMTI is a special case of HRT-MSLQ in the marriage model.
In fact, we can show that the two problems are polynomially equivalent in approximability by the following reduction:
Let $I$ be an instance of HRT-MSLQ in the marriage model, where $R$ and $H$ are respectively the sets of residents and hospitals.
For each $r_{i} \in R$, construct a man $m_{i}$ and for each $h_{j} \in H$, construct a woman $w_{j}$.
A man $m_{i}$'s preference list is that of $r_{i}$ in which a hospital $h_{j}$ is replaced by a woman $w_{j}$.
Similarly, a woman $m_{j}$'s preference list is that of $h_{j}$ in which a resident $r_{i}$ is replaced by a man $m_{i}$.
Furthermore, for each $[0,1]$-hospital $h_{j}$, we construct a man $a_{j}$ who includes only $w_{j}$ in the preference list, and we append $a_{j}$ to $w_{j}$'s preference list as the last choice.
This completes the reduction and let $I'$ be the reduced instance.

For a stable matching $M$ of $I$, define a matching $M'$ of $I'$ as $M'=\set{(m_{i}, w_{j}) \mid (r_{i}, h_{j}) \in M} \cup \set{(a_{j}, w_{j}) \mid h_{j} \mbox{ is a [0,1]-hospital unmatched in } M}$.
It is easy to see that $M'$ is stable.
We can also see that any stable matching of $I'$ can be obtained in this manner because a woman $w_{j}$ corresponding to a $[0,1]$-hospital $h_{j}$ is always matched in a stable matching.
Hence this is a one-to-one correspondence between the stable matchings of $I$ and those in $I'$.
Since the score of $M$ is the same as the size of $M'$, the correctness of the reduction follows.
\end{remark}

\begin{toappendix}
\subsection{Proof of Theorem~\ref{thm:VC-PM-gap}}\label{sec:postponed}
We obtain the NP-hardness for VC-PM claimed in the statement by transforming the following recent result for VC.
\begin{theorem}[Khot, Minzer, and Safra \cite{Khot18}]
\label{th-hardness-cv-1}
For any positive $\delta$, it is NP-hard for graphs $G=(V,E)$ to distinguish between the following two cases: 
\begin{itemize}
\switch{\setlength{\leftskip}{3mm}}{}
\item[\rm (1)] $\optvc(G) \leq (\frac{\sqrt{2}}{2}+\delta)|V|$.
\item[\rm (2)]$\optvc(G) \geq (1-\delta)|V|$. 
\end{itemize}
\end{theorem}

Let $G=(V,E)$ be a graph.
It is shown by Nemhauser and Trotter \cite{NT75} that a vertex set $V$
can be partitioned in polynomial time into three subsets $V_0$, $V_1$, and $V_{1/2}$ satisfying the following conditions. Here, $G[W]$ denotes the subgraph of $G$ induced by a subset $W\subseteq V$ and $\optvc^*(G')$ denotes the optimal value of the fractional vertex cover problem for a graph $G'$.
\begin{itemize}
\item $V_0$ is an independent set of $G$ and $V_1$ is the set of all neighbors of $V_0$, 
\item $\optvc^*(G[V_{1/2}]) = \frac{1}{2} |V_{1/2}|$,  and
\item there exists a minimum vertex cover $C$ of $G$ such that $V_1 \subseteq C \subseteq V_1 \cup  V_{1/2}$ and $C \cap V_{1/2}$ is a minimum vertex cover for $G[V_{1/2}]$.

\end{itemize}
By definition, we have $\optvc^*(G[V_{1/2}]) \leq \optvc(G[V_{1/2}])$. 
By making use of such a partition, Chleb{\'\i}k and Chleb{\'\i}kov{\'a}~\cite{CC07} constructed, given a graph $G$, the graph $\hat{G}=(\hat{V}, \hat{E})$ defined by
\begin{eqnarray*}
\hat{V}&=& \set{ v, v' | v \in V_{1/2}} \\
\hat{E}&=&\set{(v_i,v_j), (v'_i,v'_j), (v'_i,v_j), (v_i,v'_j)| (v_i,v_j) \in E_{1/2}}, 
 \end{eqnarray*}
where $v'$ is a copy of $v \in V_{1/2}$ and $E_{1/2}$ denotes the set of edges in $G[V_{1/2}]$. 
Note that $\hat{G}$ can be constructed from $G$ in polynomial time and $|\hat{V}|=2|V_{1/2}|$.
It is shown in \cite{CC07} that $\optvc(\hat{G})=2\optvc(G[V_{1/2}])$ holds
and  $\hat{G}$ has a perfect matching. 

To show the statement of Theorem~\ref{thm:VC-PM-gap}, take an arbitrary but fixed positive value $\epsilon$. Set a positive constant $\delta$ such that 
\[
\delta\leq \min \biggl\{\epsilon, \frac{(2-\sqrt{2})\epsilon}{2 + \sqrt{2}+4\epsilon }\biggr\}.
\]
\begin{\switch{myclaim}{claim}}
\label{claim-1a-}
The problem for graphs $G=(V,E)$ to distinguish between (1) and (2) in Theorem~\ref{th-hardness-cv-1}
can be reduced into the problem for graphs $\hat{G}=(\hat{V},\hat{E})$ (with perfect matchings)  to distinguish between the following two cases:
\begin{itemize}
\switch{\setlength{\leftskip}{3mm}}{}
\item[\rm (1$'$)] $\optvc(\hat{G}) \leq (\frac{\sqrt{2}}{2}+\epsilon)|\hat{V}|$. 
\item[\rm (2$'$)]$\optvc(\hat{G}) \geq (1-\epsilon)|\hat{V}|$. 
\end{itemize}
\end{\switch{myclaim}{claim}}

Recall that $V$ can be decomposed, in polynomial time, into subsets $V_0$, $V_1$, and $V_{1/2}$ satisfying the three conditions mentioned above. Let $\alpha$ and $\beta$ be numbers such that
$|V_0|=\alpha|V|$ and $|V_1|=\beta|V|$. 
By the third condition above, we have $\beta|V|\leq \tau(G)\leq (1-\alpha)|V|$.
If $\alpha>\delta$, then we can conclude that (1) in Theorem \ref{th-hardness-cv-1} holds. 
On the other hand, if $\beta > \frac{\sqrt{2}}{2}+\delta$, then we can conclude that (2) in Theorem \ref{th-hardness-cv-1} holds. 
Hence, in the rest, we assume that $\alpha \leq  \delta$ and  $\beta \leq  \frac{\sqrt{2}}{2}+\delta$. 
Let $c$ be a number with  $\optvc(G[V_{1/2}])=c|V_{1/2}|=c(1-\alpha-\beta)|V|$. Then we have $\optvc(G)=|V_1|+\optvc(G[V_{1/2}])=(\beta+c(1-\alpha-\beta))|V|$.

In case (1) of Theorem \ref{th-hardness-cv-1} holds,
we have $\optvc(G)=(\beta+c(1-\alpha-\beta))|V| \leq (\frac{\sqrt{2}}{2}+\delta) |V|$. This implies 
\begin{eqnarray*}
c&\leq&\frac{\sqrt{2}}{2}+\frac{\frac{\sqrt{2}}{2}\alpha-(1-\frac{\sqrt{2}}{2})\beta+\delta}{1-\alpha-\beta}\\
&\leq&\frac{\sqrt{2}}{2}+\frac{(\sqrt{2}+2)\delta}{
2-\sqrt{2}-4\delta} \ \leq \ \frac{\sqrt{2}}{2}+ \epsilon, 
\end{eqnarray*}
where the second inequality follows from $\alpha \leq  \delta$
and  $0\leq \beta \leq  \frac{\sqrt{2}}{2}+\delta$, while the third one follows from the definition of $\delta$.  
Since we have $\optvc(\hat{G})=2\optvc(\hat{G[V_{1/2}]})=2c|V_{1/2}|=c|\hat{V}|$,
this inequality implies $\optvc(\hat{G}) \leq (\frac{\sqrt{2}}{2}+ \epsilon)|\hat{V}|$, i.e., (1$'$) holds.
\smallskip

In case (2) of Theorem \ref{th-hardness-cv-1} holds,
we have $\optvc(G)=(\beta+c(1-\alpha-\beta))|V| \geq (1-\delta) |V|$. This implies 
\begin{eqnarray*}
c&\geq&\frac{1-\beta-\delta}{1-\alpha-\beta}\ 
= \ 1+\frac{\alpha-\delta}{1-\alpha-\beta}\\
& \geq & 1-\frac{\delta}{1-\beta}\  \geq \ 1-\delta \  \geq \ 1-\epsilon,  
\end{eqnarray*}
where the second, third, and fourth inequalities respectively follow from $\alpha\geq 0$, $\beta\geq 0$, and the definition of $\delta$. 
Therefore, we have 
$\optvc(\hat{G}) =c|\hat{V}|\geq (1- \epsilon)|\hat{V}|$, i.e., (2$'$) holds.

We complete the proof of Claim \ref{claim-1a-}, which together with Theorem~\ref{th-hardness-cv-1} implies Theorem~\ref{thm:VC-PM-gap}. 
\end{toappendix}

\begin{toappendix}
\subsection{Proof of Proposition~\ref{prop-hard--1x}} \label{sec:postponed2}
Let us now prove Proposition~\ref{prop-hard--1x}, which connects VC-PM and HRT-MSLQ.
We explain the construction of $I$ for a given graph $G=(V,F)$ with $V=\{v_1,v_2,\dots, v_n\}$. 
Let $F^*\subseteq F$ be a perfect matching of $G$.
For each edge  $(v_i, v_j)$ in $F^*$, we construct a gadget that consists of $5 u$ residents $\set{a^j_k, b^i_k,c^{ij}_k,b^j_k, a^i_k|k=1,2,\dots, u}$ and $3+2u$ hospitals $\{y^j,z^{ij}, y^i\}\cup \set{x^i_k, x^j_k|k=1,2,\dots,u}$.
Since $I$ have no other resident and hospital and $F^*$ is a perfect matching of $G$, 
$I$ consists of $5u|F^*|=2.5 un$ residents and $(3+2u)\cdot |F^*|=1.5n+un$ hospitals.
We denote by $R$ and $H$ the sets of residents and hospitals of $I$, respectively. 
Define the set $E$ of acceptable resident-hospital pairs in $I$  by $E=\bigcup_{(v_i,v_j) \in F^*}E_{ij} \cup E_{\rm out}$, where 
\begin{eqnarray*}
E_{ij}&=&\set{
(b^i_k, y^j),  (b^j_k, y^i), (c^{ij}_k, z^{ij})
\mid k=1,2,\dots, u}\\
&& \hspace*{.3cm} \cup \ 
\set{(a^p_k, y^p), 
(b^p_k, z^{ij}),  (b^p_k, x^p_k), 
(c^{ij}_k,y^p)
\mid  p \in \{i,j\},  ~k=1,2,\dots, u}\\E_{\rm out}&=&\set{(b^{p}_k, y^{q}), (b^{q}_k, y^{p})| (v_{p}, v_{q})\in F\setminus F^*, ~k=1,2,\dots, u}.
\end{eqnarray*}
We also define $E_{\rm in}=\bigcup_{(v_i,v_j) \in F^*}E_{ij}$. 
By definition,  $E_{\rm in}$ consists of pairs in the same  gadgets, while $E_{\rm out}$ consists of pairs connecting different gadgets. 
We also see that 
$|E_{\rm in}|=11u|F^*|=5.5un$ and 
$|E_{\rm out}|=2u(|F|-|F^*|)=2u|F|-un$. 
The set $E$, together with preference lists is given in Figure \ref{fig:gadget}, 
where 
the quotas are $[\ell, u]$ for all hospitals and $k$ takes values $1,2,\dots, u$.
In the figure, for $i=1, 2,\dots, n$ and $k=1,2,\dots, u$,  ``[$\Gamma_{\rm out}(b^i_k)$]'' denotes a strict preference that arranges the hospitals in
$\set{y_j|(b^i_k, y^j)\in E_{\rm out}}$ in increasing order of indices $j$, and 
for $i=1,2,\dots, n$, ``[$\Gamma_{\rm out}(y^i)$]'' denotes a strict preference that arranges the residents in
$\set{b^j_k|(b^i_k, y^j)\in E_{\rm out}}$ defined by indices $(j,k)$
in lexicographic order of indices $(j,k)$.

\begin{figure}[h]
\begin{center}
\renewcommand\arraystretch{1.4}
\begin{tabular}{llllllllllllllllll}
&$a^j_k$: & $y^j$   &\hspace{10mm} & $x^i_k$ : $b^i_k$&  \\
&$b^i_k$: & $(y^j~z^{ij})~[\Gamma_{\rm out}(b^i_k)]~x^i_k$   &\hspace{10mm} & $y^j$ : $c^{ij}_1\cdots c^{ij}_{u}~~ b^i_1\cdots b^i_u~[\Gamma_{\rm out}(y^j)] ~a^j_1\cdots a^j_u$&  \\
&$c^{ij}_k$: & $z^{ij}~(y^i~y^j)$  &\hspace{10mm} & $z^{ij}$ : $(b^i_1\cdots b^i_u~b^j_1\cdots b^j_u)~~c^{ij}_1\cdots c^{ij}_u$&   \\
&$b^j_k$: & $(y^i~z^{ij})~[\Gamma_{\rm out}(b^j_k)]~x^j_k$  &\hspace{10mm} & $y^i$ : $c^{ij}_1\cdots c^{ij}_{u}~~ b^j_1\cdots b^j_u~[\Gamma_{\rm out}(y^i)] ~a^i_1\cdots a^i_u$&  \\
&$a^i_k$: & $y^i$   &\hspace{10mm} & $x^j_k$ : $b^j_k$&  
\end{tabular}
    \caption{Acceptable resident-hospital pairs and preference lists of $I$}
    \label{fig:gadget}
\end{center}
\end{figure}

We show that this instance $I$ of HRT-MSLQ
satisfies the two conditions (a) and (b) in the theorem.  
We first show that, for any vertex cover $C\subseteq V$ of $G$, we can construct a stable matching $M$ of $I$ with $s(M)=(1.5+\theta)|V|-\theta|C|$. 
This implies $\opt(I)\geq (1.5+\theta)|V|-\theta\optvc(G)$.

Let $C\subseteq V$ be a vertex cover in $G$.
For each edge $(v_i, v_j)\in F^*$, let us define a matching $M_{ij} \subseteq E_{ij}$ in $I$  by 
\begin{equation*}
M_{ij}=\begin{cases}
 \{(b_k^i, y^j), (c^{ij}_k, z^{ij}), (b_k^j, y^i) \mid k=1, \dots , u \} & \mbox{if } C\cap \{v_i,v_j\}=\{v_i,v_j\}\\
 \{(b_k^i, x_k^i), (c^{ij}_k, y^{j}), (b_k^j, z^{ij}), (a_k^i,y^i) \mid k=1, \dots , u \} &\mbox{if }C\cap \{v_i,v_j\}=\{v_j\} \\
 \{(a_k^j,y^j), (b_k^i, z^{ij}), (c^{ij}_k, y^{i}), (b_k^j, x^{j}_k)  \mid k=1, \dots , u \} &\mbox{if }C\cap \{v_i,v_j\}=\{v_i\}, 
\end{cases}
\end{equation*}
and let $M=\bigcup_{(v_i, v_j)\in F^*}M_{ij}$.
Note that $C\cap \{v_i,v_j\}\not=\emptyset$ holds, because $C$ is a vertex cover in $G$. 
It is not difficult to see that a matching $M$ is stable. In fact, no pair in $E_{\rm in}$ blocks $M$. 
To see that no pair in $E_{\rm out}$ blocks $M$, suppose contrary that $(b^i_k, y^j)\in E_{\rm out}$ is a blocking pair of $M$. 
Then $M$ contains pairs $(b^i_k, x_k^i)$ and $(a_k^j,y^j)$.
However, this implies that $C$ contains neither $v_i$ nor $v_j$, which contradicts  that $C$ is a vertex cover in $G$.
Denote by $s(M_{ij})$ the total score of the hospitals in the gadget of $(v_i, v_j)\in F^*$.
Then, $s(M_{ij})=3$ if $C\cap \{v_i,v_j\}=\{v_i,v_j\}$ and $3 +\theta$ otherwise.
Thus, we have $s(M)=3(|C|-0.5|V|)+
(3 +\theta) (|V|-|C|)=(1.5+\theta)|V|-\theta|C|$. 

Let us next show that, for any stable matching $M$ of $I$, we can construct a vertex cover $C$ in $G$ with
$s(M)\leq (1.5+\theta)|V|-\theta|C|$. 
This shows (a) and $\opt(I)\leq (1.5+\theta)|V|-\theta\optvc(G)$, 
which together with the other inequality above implies (b).
In order to prove it, 
we provide several structural properties on a stable matching in $I$.

\begin{\switch{myclaim}{claim}}
\label{claim-m1}
For any stable matching $M$ in $I$ and any $(v_i,v_j)\in F^*$, the following statements~hold. 
\begin{itemize}
\item[\rm (i)]Each of residents $b_k^i$,  $b_k^j$, and $c_k^{ij}$ is assigned to some hospital in $M$ for any $k$. 
\item[\rm (ii)] Each of hospitals $y^i$, $y^j$, and $z^{ij}$ is assigned $u$ residents in $M$. 
\end{itemize}
\end{\switch{myclaim}{claim}}
\begin{proof}
A resident $b^i_k$ is matched in $M$ since otherwise $(b^i_k, x^i_k)$ blocks $M$ while $M$ is stable.
Similarly, $b^i_k$ is matched in $M$. Suppose, to the contrary, that $c^{ij}_k$ is unmatched.
By the stability of $M$, each of $z^{ij}$, $y^i$, $y^j$ is assigned $u$ residents which are no worse  than $c^{ij}$.
Therefore, residents in $\set{c^{ij}_h, b^i_h, b^j_h|h=1,2,\dots,u}\setminus \{c^{ij}_k\}$ should fill $3u$ seats of $z^{ij}$, $y^i$, $y^j$, a contradiction. Thus, the statement (i) is proved. The statement (ii) is shown by a similar argument.
\switch{\qed}{}\end{proof}

\begin{\switch{myclaim}{claim}}
\label{claim-m2}
Let $M$ be a stable matching  in $I$, and let  
$(v_i,v_j)$ be an edge in $F^*$. 
If some resident $b_k^i$ is assigned to $x_k^i$ or $y^p$ with $p \not=j$ in $M$, 
then the following statements hold.
\begin{itemize}
\item[\rm (i)] Residents assigned to $y^j$ in $M$ are contained in $\set{b_h^i, c^{ij}_h | h=1, \dots , u}$. 
\item[\rm (ii)] 
Resident $b_h^j$ is assigned to $y^i$ or $z^{ij}$ in $M$ for any $h$. 
\end{itemize}
\end{\switch{myclaim}{claim}}
\begin{proof}
The hospital $y^j$ is assigned $u$ residents from $c^{ij}_1\cdots c^{ij}_{u}~~ b^i_1\cdots b^i_u$ 
since otherwise $(b^i_k, y^j)$ blocks $M$. Thus, (i) holds.
To show (ii), suppose, to the contrary, that some $b_h^j$ is assigned to neither $y_i$ and $z^{ij}$.
Then, $y^i$ is assigned $u$ residents from $c^{ij}_1\cdots c^{ij}_{u}~~ b^j_1\cdots b^j_u$, since otherwise $(b^j_h, y^i)$ blocks $M$. In addition, $z^{ij}$ is assigned $u$ residents by Claim~\ref{claim-m1}.
Therefore, residents in $\set{c^{ij}_q, b^i_q, b^j_q|q=1,2,\dots,u}\setminus \{b^i_k\}$ should fill $3u$ seats of $z^{ij}$, $y^i$, $y^j$, a contradiction.
\switch{\qed}{}\end{proof}

\begin{\switch{myclaim}{claim}}
\label{claim-m3}
Let $M$ be a stable matching  in $I$, and let  
$(v_i,v_j)$ be an edge in $F^*$. If the hospital $y^i$ is assigned some $a_k^i$ or $b_k^p$ with $p \not=j$ in $M$, 
then the following statements hold.
\begin{itemize}
\item[\rm (i)] Resident $b_h^j$ is assigned to $y^i$ or $z^{ij}$ in $M$ for any $h$. 
\item[\rm (ii)] 
Residents assigned to $y^j$ in $M$ are contained in $\{b_h^i, c^{ij}_h \mid h=1, \dots , u\}$. h
\end{itemize}
\end{\switch{myclaim}{claim}}
\begin{proof}
Note that, $b^p_k$ with $p\neq j$ (if assigned to $y^i$) appears in $[\Gamma_{\rm out}(y^i)]$ in the list of $y^i$,
and hence worse than $b^j_h$ for any $h$.
Then, for any $h$, resident $b^j_h$ is assigned to $y_i$ or $z^{ij}$ since otherwise $(b^j_h, y^i)$ blocks $M$. Thus, (i) holds.
To see (ii), suppose contrary that $y^j$ is assigned a resident not in $\{c^{ij}_1\cdots c^{ij}_{u}~~ b^i_1\cdots b^i_u\}$ .
Then, for any $h$, resident $b^i_h$ is assigned to $y^i$ or $z^{ij}$ since otherwise $(b^i_h, y^j)$ blocks $M$.
By Claim~\ref{claim-m1}, for any $h$, resident $c^{ij}_h$ is assigned to $y^i$, $y^j$, or $z^{ij}$.
Therefore, $3u$ residents in $R\coloneqq\set{c^{ij}_q, b^i_q, b^j_q|q=1,2,\dots,u}$ are all assigned to $y^i$, $y^j$, and $z^{ij}$, while $y^i$ is assigned a resident outside $R$, contradicting that upper quotas of hospitals are all $u$.
\switch{\qed}{}\end{proof}
For a stable matching $M$ in $I$ and an edge $f=(v_i,v_j) \in F^*$, let $\alpha(f)=\alpha_i+\alpha_j$ and $\beta(f)=\beta_i+\beta_j$, where $\alpha_p=|\{(b_k^p,x_k^p) \in M \mid k=1, \dots , u\}|$
and $\beta_p=|\{(a_k^p,y^p) \in M\mid k=1, \dots , u\}|$ for $p=i,j$. 
Then the above claims imply the following property. 

\begin{\switch{myclaim}{claim}}
\label{claim-m4}
Let $M$ be a stable matching in $I$, and let  
$(v_i,v_j)$ be an edge in $F^*$. 
Then, the following statements hold.
\begin{itemize}
\item[\rm (i)]If $\alpha_i$ or $\beta_i$ is positive, then $\alpha_j=\beta_j=0$. 

\item[\rm (ii)] 
 If $\alpha_j$ or $\beta_j$ is positive, then $\alpha_i=\beta_i=0$. 
\end{itemize}
\end{\switch{myclaim}{claim}}
\begin{proof}
Each statement follows from Claims~\ref{claim-m2} and \ref{claim-m3}.
\switch{\qed}{}\end{proof}

In the following,  we strengthen Claim \ref{claim-m4} to 
\[
\mbox{$\alpha_i=\beta_i$, $\alpha_j=\beta_j$, and at most one of them is positive}, 
\]
which is (under Claim \ref{claim-m4}) equivalent to the fact that $\alpha(f)=\beta(f)$ for any edge $f \in F^*$.

For a stable matching $M$ in $I$, let $M_{\rm in}= M \cap E_{\rm in}$ and $M_{\rm out}= M \cap E_{\rm out}$. Let  
$(v_i,v_j)$ be an edge in $F^*$. 
By Claim \ref{claim-m2},   
if $b_k^i$ (for some $k$) is incident to a pair in $M_{\rm out}$, then none of  $b_h^j$ ($h=1, \dots , u$) and $y^j$ is incident to a pair in $M_{\rm out}$. 
Similarly, 
by Claim \ref{claim-m3}, if $y^i$  is incident to a pair in $M_{\rm out}$, then none of  $b_h^j$ ($h=1, \dots , u$) and  $y^j$ is incident to a pair  in $M_{\rm out}$. 
Let us  consider a directed graph $D=(F^*,M^*_{\rm out})$ which is obtained from the graph $(R, H; M_{\rm out})$  
by orienting  all the pairs in $M_{\rm out}$ from  residents to hospitals and  identifying each gadget into a single vertex.
There are parallel arcs between two vertices in $D$ if $M_{\rm out}$ contains multiple pairs connecting the corresponding two gadgets. By analyzing this directed graph $D$, we show the following claim.

\begin{\switch{myclaim}{claim}}
\label{claim-m5}
For a stable matching $M$ in $I$, 
we have $\alpha(f)=\beta(f)$ for any edge $f \in F^*$. 
\end{\switch{myclaim}{claim}}
\begin{proof}
Let $f=(v_i,v_j)$ be an edge in $F^*$. 
By Claim \ref{claim-m4} we assume  without loss of that $\alpha_j=\beta_j=0$. 
Then Claims \ref{claim-m1}, \ref{claim-m2}, and \ref{claim-m3} imply
\[
\alpha_i+|\{(b_k^i,y^p) \in M \mid p \not=j, k=1, \dots , u\}|=
\beta_i+|\{(b_k^p,y^i) \in M \mid p\not=j,  k=1, \dots , u\}|. 
\]
Note that  $|\{(b_k^i,y^p) \in M \mid p \not=j, k=1, \dots , u\}|$ and $|\{(b_k^p,y^i) \in M \mid p\not=j,  k=1, \dots , u\}|$ are  out- and in-degree of vertex $f$ in $D$, respectively.
Thus the statement in the claim holds if and only if any node in the directed graph $D=(F^*,M^*_{\rm out})$ has  the same in- and out-degree.  

Suppose, to the contrary, that $D$ has a vertex whose in-degree and out-degree are different. 
Then there exists a simple directed path $P$ from $f$ to $g$ in $D$ such that $\alpha(f) < \beta(f)$ and $\alpha(g) > \beta(g)$.  
Let $f=(v_{i_0},v_{j_0})$ and $g=(v_{i_d}, v_{j_d})$, where $d$ denotes the length of $P$, and let $(b_{k_0}^{i_0},y^{i_1})$,$(b_{k_1}^{i_1}, y^{i_2}), \dots, (b_{k_{d-1}}^{i_{d-1}}, y^{i_{d}})$ be a path in $(R, H; M_{\rm out})$ corresponding to $P$. 
By  $\alpha(f) < \beta(f)$ and $\alpha(g) > \beta(g)$, 
$(R, H; M_{\rm out})$ also contains 
two pairs $(a_{r}^{i_0},y^{i_0})$ and $(b_{t}^{i_d},x_t^{i_d})$ for some $r$ and $t$. 
Note that hospital $y^{i_0}$ prefers $b_{k_1}^{i_1}$ to $a_{r}^{i_0}$. 
Since $(b_{k_1}^{i_1}, y^{i_0})$ is not  a blocking pair of  $M$, $b_{k_1}^{i_1}$ prefers $y^{i_2}$ to $y^{i_0}$. 
This means that $y^{i_1}$ prefers $b_{k_2}^{i_2}$ to $b_{k_0}^{i_0}$, since preference orders on $E_{\rm out}$ are defined lexicographically.  By repeatedly applying this argument, we can conclude that $y^{i_{d-1}}$ prefers $b_{k_d}^{i_d}$ to $b_{k_{d-2}}^{i_{d-2}}$. 
Note that $(y^{i_{d-1}},  b_{k_{d-2}}^{i_{d-2}})$ and $(b_{t}^{i_d},x_t^{i_d})$ are contained in $M$.  
Since $b_{t}^{i_d}$ prefers $y^{i_{d-1}}$ to $x_t^{i_d}$, 
$M$ contains a blocking pair $(b_{t}^{i_d}, y^{i_{d-1}})$, which contradicts the stablity of $M$. 
This completes the proof. 
\switch{\qed}{}\end{proof}

It follows from Claim \ref{claim-m5} that stable matchings in $I$ provide vertex covers in $G$ as follows.

\begin{\switch{myclaim}{claim}}
\label{claim-m6}
For a stable matching $M$ in $I$, 
$C_M=\{ v_i \in V \mid \alpha_i =0\}$ is a vertex cover in $G$. 
\end{\switch{myclaim}{claim}}
\begin{proof}
For any edge $(i,j) \in F$,  Claim \ref{claim-m4} implies that at least of $i$ and $j$ is contained in $C_M$. 
Let us then consider an edge $(i,j) \in F \setminus F^*$. 
Assume that  $\alpha_i$ and $\alpha_j$ are both positive.  
Then by definition, $M$ contains a pair   $(b^i_k, x_k^i)$ for some $k$.  
Since $\beta_j>0$ follows from Claim \ref{claim-m5},  $M$ also contains  $(a^j_h, y^j)$ for some $h$. 
However, these imply that $(b^i_k, y^j)$ is a blocking pair of $M$, a contradiction.
Hence at least one of $\alpha_i$ and $\alpha_j$ is zero, which completes the lemma. 
\switch{\qed}{}\end{proof}

By Claim \ref{claim-m6}, a stable matching $M$ in $I$ has score 
\begin{eqnarray*}
s(M) &=& \sum_{f \in F^*} \left(3+ \frac{\alpha(f)}{\ell}\right)\\
&\leq & 3|F^*|+ \theta(|V|-|C_M|)\ \leq \  (1.5+\theta)|V|-\theta|C_M|, 
\end{eqnarray*}
where the second inequality follows from the fact that $\alpha(f)\leq u$.
This immediately implies that 
$\opt(I) \leq (1.5+\theta)|V|-\theta\optvc(G)$, which completes the proof of Proposition \ref{prop-hard--1x}. 
\end{toappendix}

\section{Conclusion}
There remain several open questions and future research directions for HRT-MSLQ.
Clearly, it is a major open problem to close a gap between the upper and lower bounds on the approximation factor for each scenario.
Next, as mentioned in Remark \ref{rem:3.1}, {\sc Triple Proposal} is not strategy-proof in general.
It would be an important task to clarify approximability of HRT-MSLQ when we restrict ourselves to strategy-proof (instead of polynomial-time) algorithms.
Finally, in our setting, the score of each hospital is piecewise linear in the number of assignees.
A possible extension is to consider more general functions, such as convex ones.

\subsubsection*{Acknowledgements}
The authors would like to thank the anonymous reviewers for reading the submitted version carefully and giving useful comments, which improved the quality of the paper.
This work was partially supported by the joint project of Kyoto University and Toyota Motor Corporation, titled ``Advanced Mathematical Science for Mobility Society''.
The first author is supported by JSPS KAKENHI Grant Numbers JP20H05967, JP19K22841, and JP20H00609.
The second author is supported by JSPS KAKENHI Grant Number JP20K11677.
The last author is supported by JSPS KAKENHI Grant Number JP18K18004 and JST, PRESTO Grant Number JPMJPR212B.
%
%
%
\bibliographystyle{plain}
\bibliography{main}

\end{document}